\documentclass{article}
\usepackage[utf8]{inputenc}

\usepackage[margin=1in]{geometry}

\usepackage{float,graphicx,verbatim,fullpage,hyperref,amssymb,amsmath,amsthm,enumerate,multicol, xspace,xcolor,mathtools,thmtools,thm-restate,cleveref,xspace,tikz,caption,subcaption,wasysym, enumitem}
\usepackage[margin=1in]{geometry}
\usepackage{algorithm, cite}
\usepackage[noend]{algpseudocode}
\usepackage{enumitem}
\usepackage{comment}

\usetikzlibrary{calc}

\providecommand{\nnreals}{\mathbb{R}_{\geq 0}}

\usepackage{mleftright}
\usepackage{hyperref}
    \hypersetup{colorlinks=true, linkcolor=blue, filecolor=magenta, urlcolor=blue, citecolor=red}

\usetikzlibrary{arrows.meta}
\tikzset{>={Latex[width=1.5mm,length=1.5mm]}}

\newfloat{procedure}{htbp}{loa}
\floatname{procedure}{Procedure}

\def\R{\mathbb{R}}
\def\Q{\mathbb{Q}}

\def\Z{\mathbb{Z}}

\newcommand{\cP}{\mathcal{P}}

\newcommand{\cD}{\mathcal{D}}

\newcommand{\cJ}{\mathcal{J}}

\newcommand{\opt}{\textsf{OPT}}
\newcommand{\obj}{\textsf{OBJ}}

\newcommand{\cost}{\textsf{cost}}
\def\ep{\varepsilon}
\def\tO{\tilde{O}}

\def\len{\operatorname{length}}

\newtheorem{theorem}{Theorem}[section]

\newtheorem{lemma}[theorem]{Lemma}
\newtheorem{claim}[theorem]{Claim}
\newtheorem{corollary}[theorem]{Corollary}

\theoremstyle{definition}
\newtheorem{definition}[theorem]{Definition}

\providecommand{\email}[1]{\href{mailto:#1}{\nolinkurl{#1}\xspace}}

\def\final{1}  
\def\iflong{\iffalse}
\ifnum\final=0  
\newcommand{\ynote}[1]{{\color{blue}[{\small Young-San: \bf #1}]\marginpar{\color{red}*}}}
\newcommand{\nnote}[1]{{\color{red}[{\small Nithish: \bf #1}]\marginpar{\color{red}*}}}
\newcommand{\enote}[1]{{\color{blue}[{\small Elena: \bf #1}]\marginpar{\color{red}*}}}
\newcommand{\todo}[1]{{\color{red}[{ TODO: \bf #1}]\marginpar{\color{red}*}}}
\else 
\newcommand{\ynote}[1]{}
\newcommand{\nnote}[1]{}
\newcommand{\enote}[1]{}
\newcommand{\todo}[1]{}
\fi

\usepackage{environ}

\newcommand{\eps}{\varepsilon}

\def\*#1{\mathbf{#1}}
\def\+#1{\mathcal{#1}}

\NewEnviron{problem}[1]{
	\begin{center}\fbox{\parbox{6in}{
				{\centering\scshape #1\par}
				\parskip=1ex
				\everypar{\hangindent=1em}
				\BODY
}}\end{center}}

\newcommand{\poly}{\ensuremath{\mathsf{poly}}}
\newcommand{\polylog}{\ensuremath{\mathsf{polylog}}}

\makeatletter
\newcommand*{\inlineequation}[2][]{
  \begingroup
    \refstepcounter{equation}
    \ifx\\#1\\
    \else
      \label{#1}
    \fi
    \relpenalty=10000 
    \binoppenalty=10000 
    \ensuremath{
      #2
    }
    ~\@eqnnum
  \endgroup
}
\makeatother
\newcommand{\gsf}{\textsc{Group Steiner Forest}\xspace}

\newcommand{\pwsul}{\textsc{Pairwise Weighted Spanner}\xspace}
\newcommand{\rsp}{\textsc{Restricted Shortest Path}\xspace}
\newcommand{\rcsp}{\textsc{Resource-constrained Shortest Path}\xspace}

\newcommand{\epsvec}{\mathcal{E}}

\newcommand{\ubbs}{\textsc{Unit-demand Buy-at-bulk Spanner}\xspace}
\newcommand{\bbs}{\textsc{Buy-at-bulk Spanner}\xspace}

\newcommand{\ssbbs}{\textsc{Single-source Buy-at-bulk Spanner}\xspace}

\newcommand{\mslc}{\textsc{Minimum Density Steiner Label Cover}\xspace}
\newcommand{\bdgt}{\textsc{Dis}\xspace}
\newcommand{\res}{\textsc{Len}\xspace}
\newcommand{\load}{\textsc{Load}\xspace}
\newcommand{\dem}{\textsc{Dem}\xspace}

\newcommand{\upc}{\text{upfront cost}\xspace}
\newcommand{\ppc}{\text{pay-per-use cost}\xspace}

\newcommand{\maxn}{\mathcal{N}}

\DeclareMathOperator{\siign}{sgn}
\newcommand{\maxni}{\mathcal{N}_i }

\newcommand{\lenup}{\textsc{Max length}\xspace}
\newcommand{\lenlow}{ \textsc{Min length}\xspace}

\newcommand{\rcjt}{\textsc{Minimum-density Distance-constrained Junction Tree}\xspace}

\newcommand{\rrcjt}{\textsc{$\theta$-relaxed Minimum-density Distance-constrained Junction Tree}\xspace}

\title{Directed Buy-at-Bulk Spanners \\}

\author{
Elena Grigorescu\thanks{Purdue University.
 E-mail: \email{elena-g@purdue.edu}. Supported in part by NSF CCF-1910659, NSF CCF-1910411, and NSF CCF-2228814.
}
 \and
 Nithish Kumar\thanks{Purdue University. E-mail: \email{kumar410@purdue.edu}. Supported in part by NSF CCF-1910411, and NSF CCF-2228814.}
 \and
 Young-San Lin\thanks{Melbourne Business School. 
 E-mail: \email{y.lin@mbs.edu}.}
}

\begin{document}

\maketitle
\begin{abstract}

We present a framework that unifies directed buy-at-bulk network design and directed spanner problems, namely, \emph{buy-at-bulk spanners}. 
The goal is to find a minimum-cost routing solution for network design problems that capture economies at scale,  while satisfying demands and distance constraints for terminal pairs. 
A more restricted version of this problem was shown to be $O(2^{{\log^{1-\ep} n}})$-hard to approximate, where $n$ is the number of vertices, under a standard complexity assumption, due to Elkin and Peleg (Theory of Computing Systems, 2007). 

Our results for buy-at-bulk spanners are the following.
\begin{enumerate}
    \item When the edge lengths are \emph{integral} with magnitude \emph{polynomial} in $n$ we present:
        \begin{enumerate}
            \item An $\tO(n^{4/5 + \ep})$-approximation polynomial-time randomized algorithm for \emph{uniform} demands.
            \item An $\tO(k^{1/2 + \ep})$-approximation polynomial-time randomized algorithm for general demands, where  $k$ is the number of terminal pairs. This can be improved to an $\tO(k^{\ep})$-approximation algorithm for the single-source problem.
        \end{enumerate}

    \item When the edge lengths are {\em rational} and {\em well-conditioned}, we present an $\tO(k^{1/2 + \ep})$-approximation polynomial-time randomized algorithm which may slightly violate the distance constraints. The result can be improved to an $\tO(k^{\ep})$-approximation algorithm for the single-source problem.
\end{enumerate}

To the best of our knowledge, these are the first sublinear factor approximation algorithms for directed buy-at-bulk spanners. Furthermore, these results hold even when we allow the edge lengths to be \emph{negative}, unlike the previous literature for spanners. Our approximation ratios match the state-of-the-art ratios in special cases, namely, buy-at-bulk network design by Antonakopoulos (WAOA, 2010) and weighted spanners by Grigorescu, Kumar, and Lin (APPROX 2023).

Our results are based on approximation algorithms for the following two problems that are of independent interest: \emph{minimum-density distance-constrained junction trees} and \emph{resource-constrained shortest path with negative consumption}. 

In the minimum-density distance-constrained junction tree problem, the goal is to find a collection of routes that share the same vertex, such that the ratio of the cost to the number of terminal demands satisfied is minimized. Our framework is an extension of the notion of {\em minimum-density junction trees} used for approximating Steiner forests by Chekuri, Even, Gupta, and Segev (SODA 2008, TALG 2011), and {\em pairwise spanners} by Chlamt{\'a}{\v{c}},  Dinitz, Kortsarz, and Laekhanukit (SODA 2017, TALG 2020). Our proposed general framework accommodates both buy-at-bulk costs and distance constraints.

In the resource-constrained shortest path problem with negative consumption, the goal is to find a path with minimum cost within a multi-dimensional resource budget. Under mild assumptions, our framework is an FPTAS extension of the resource-constrained shortest path problem by Horvath and Kis (Optimization Letters 2018) and the restricted shortest path problem (where the resource dimension is one) by Hassin (Math of OR 1992) and Lorenz and Raz (OR Letters 2001). Our result allows for negative resource consumption, unlike the previous literature.

\end{abstract}

\newpage

\section{Introduction}
A core network connectivity problem is the {\em buy-at-bulk} problem \cite{salman2001approximating,cekp,antonakopoulos2009,awerbuch1997buy,talwar2002single,gupta2003simpler,charikar2005non,chekuri2010approximation,meyerson2008cost}, in which the goal is to route resources between pairs of source and destination locations.
Each pair has an associated demand, i.e., the load of the resources to be delivered. To model economies of scale, each edge in the network is associated with a \emph{cost} given by a subadditive function for the total load of resources delivered through an edge. A feasible solution to the problem is a collection of routes for each demand.
The goal is to find a feasible solution that minimizes the overall cost.

The \emph{spanner} problem \cite{bhattacharyya2012transitive,ElkinP07,dinitz2011directed,berman2013approximation,chlamtavc2020approximating,BodwinW16,dinitz2016approximating,gkl2023,elkin1999client,Kortsarz2001OnTH}, on the other hand, is a fundamental network connectivity problem where each edge is assigned a \emph{length}, and each terminal pair is associated with a \emph{distance budget}. The goal is to find a minimum-size subgraph such that each pair is connected within its target distance.

Both the buy-at-bulk and spanner problems have been well-studied problems in their own right as they have a plethora of applications in theory and practice. One limitation of the buy-at-bulk problem is that it does not account for distance constraints. For example, although the buy-at-bulk cost formulation captures the economies of scale excellently in communication networks such as the Internet, it does not account for latency. On the other hand, spanners can capture distance-constrained connectivity but do not account for economies of scale cost when the terminal pairs have various demands.

As a specific example, consider a situation where the city council wishes to modernize the city transportation network in order to minimize the carbon footprint of commuters. This could be captured by a buy-at-bulk formulation: if more commuters use a single transportation link, say a train, then it is natural to assume that the per-commuter carbon footprint will decrease. However, commuters also have their self-interest in mind, which may be a budget for transportation. How to minimize the carbon footprint without exceeding the commuters' budget?

Therefore, complex network design problems often need to be modeled  by buy-at-bulk costs, while also satisfying spanner-like distance constraints, which leads to the following question:

\begin{center}
      \emph{How to solve buy-at-bulk network design problems while also satisfying spanner-like distance constraints?}
\end{center}

Furthermore, typical spanner problems have only dealt with positive edge lengths. However, one encounters natural applications in which the resource cost may be modeled as being \emph{negative}, i.e., a resource gain.\footnote{Consider for instance, common situations like refueling a truck in a gas station - this can easily be captured by an edge with negative fuel consumption.} 
Returning to our running example, the commuter may drive an electric car as part of his commute. This car gets recharged when going downhill, which may be captured by a gain in money when traveling on such an edge since recharging means spending less on a charging station.  
Hence, we may further ask:

\begin{center}

  \emph{How to solve buy-at-bulk network design problems with distance constraints and negative edge lengths?}

\end{center}

Motivated by these questions, we study a general multi-commodity problem in directed graphs, namely, the \emph{buy-at-bulk spanner} problem. We obtain the first results for this general formulation, which are also comparable with the best-known results from the literature on the two individual problems.

We now proceed to formalize the buy-at-bulk spanner problem.

\paragraph{Buy-at-bulk spanners.} In the buy-at-bulk spanner problem, we are given a directed simple graph $G=(V,E)$ with $n$ vertices, and a set of $k$ terminal pairs $D \subseteq V \times V$. 
Each pair $(s,t) \in D$ is associated with a \emph{distance budget} given by the function $\bdgt: D \to \R$ and a \emph{demand} given by the function $\dem: D \to \Z_{> 0}$.
When $\dem(s,t) = 1$ for all $(s,t) \in D$, we say that the problem has \emph{unit} demands. 
Each edge $e \in E$ is associated with a \emph{cost} given by a subadditive function $f_e: \nnreals \to \nnreals$, satisfying $f_e(x+y) \le f_e(x) + f_e(y)$ for all $x, y \ge 0$, and a length $\ell_e \in \R$. 
We note that the length can be \emph{negative}, which captures the notion of {\em  gain} while using an edge. The distance budget can also be negative.
We further assume that there are \emph{no negative length cycles} in $G$.
A feasible solution to the problem is a collection of paths $\cP := \{p(s,t)\}_{(s,t) \in D}$ where $p(s,t)$ is the set of edges in the directed $s \leadsto t$ path satisfying the distance requirement $\sum_{e \in p(s,t)} \ell_e \le \bdgt(s,t)$. Let the \emph{load} of edge $e$ be $\load(e):=\sum_{(s,t) \in D: e \in p(s,t)}\dem(s,t)$.
The goal is to find a feasible solution that minimizes the objective $\sum_{e \in E}f_e(\load(e))$.

The buy-at-bulk spanner problem captures a wide range of network connectivity problems that are motivated by common scenarios, such as product delivery, transportation, electricity, and internet construction.

This general formulation has been studied under many variants: without distance constraints, it is equivalent to the \emph{buy-at-bulk} problem \cite{salman2001approximating}; when the edge cost is a fixed value once used, it captures the \emph{weighted spanner} problem \cite{gkl2023}. The weighted spanner problem is a generalization of spanners, distance preservers, and Steiner forests, which have found applicability in various domains such as multi-commodity network design \cite{gupta2003approximation,fleischer2006simple}, approximate shortest paths \cite{DorHZ00, Elkin05, BaswanaK10}, distributed computation \cite{Awerbuch, PelegS89}, and routing schemes \cite{PelegU89a,CowenW04,RodittyTZ08,PachockiRSTW18}. 

\paragraph{A two-metric-based buy-at-bulk formulation.} Previous work \cite{cekp,chekuri2010approximation,antonakopoulos2009} reduces the buy-at-bulk problem to the \emph{two-metric network design} problem, with only a constant factor loss in the approximation guarantee.

In this problem, each edge $e \in E$ has a one-time setup cost $\sigma(e)$ and a pay-per-use cost $\delta(e)$.
The objective is to minimize the cost
\begin{equation} \label{obj:2m}
    \sum_{e \in \cup_{(s,t) \in D} p(s,t)} \sigma(e) + \sum_{(s,t) \in D} \sum_{e \in p(s,t)} \delta(e) \cdot \dem(s,t).
\end{equation}

\begin{lemma}\cite{chekuri2010approximation}
    Given any feasible solution with objective value $\obj_{BB}$ for the buy-at-bulk problem, there exists an instance of the two-metric network design problem that has a feasible solution with objective value $\obj_{2M}$, such that $\obj_{2M} \le \obj_{BB} \le (2 + \ep)\obj_{2M}$.
\end{lemma}

We note that adding distance constraints does not affect the reduction. Let $\cD \subseteq \R$ be the domain for the distance of each edge. We consider the following problem throughout the paper with different options of $\cD$.

\begin{restatable}{definition}{defbbs} 
\label{def:bbs}
    
    \bbs on $\cD$ 
    
    \textbf{Instance}: A directed graph $G = (V, E)$, where
    \begin{itemize}
        \item  each edge $e \in E$ has an \upc $\sigma: E \to \Q_{\ge 0}$, a \ppc $\delta: E \to \Q_{\ge 0}$, and a distance $\ell_e \in \cD$. Furthermore, we assume that there are no negative cycles induced by $\{\ell_{e}\}_{e \in E}$.
        
        \item We are given a set $D \subseteq V \times V$ of ordered pairs, where each pair has a demand captured by the function $\dem: D \to \Z_{\ge 0}$ and a distance budget captured by the function $\bdgt: D \to \R \setminus \{0\}$ \footnote{One workaround if we want to set a specific distance constraint as $0$ is to set it to a small number that is close enough to $0$ like say $10 ^ {- c}$ (where $c>0$) while ensuring the problem is well-conditioned.}. Furthermore, we assume that there exists an $s \leadsto t$ path that satisfies the distance constraint for each $(s,t) \in D$.
    \end{itemize}   
    
    \textbf{Objective}: Find a collection of $s \leadsto t$ paths $\cP:=\{p(s,t)\}_{(s,t) \in D}$ such that the cost \eqref{obj:2m} is minimized and the distance requirement $\sum_{ e \in p(s,t)} \ell_e \le \bdgt(s,t)$ is satisfied for each $(s,t) \in D$.
\end{restatable}

The performance of an approximation algorithm is measured by the \emph{approximation ratio}, the ratio between the cost of the approximate solution and the optimal solution. For \bbs, we consider two domains for edge lengths, $\cD = [\poly(n)]_{\pm} := \{j \in \Z \mid |j| \le \poly(n)\}$ and $\cD = \R$. When the problem has unit demands, \bbs is termed \ubbs.
In a special case of \bbs where the source vertex $s \in V$ is fixed, we call this problem \ssbbs. For notation convenience, we have $D \subseteq \{s\} \times V$. We say that $s$ is the \emph{root} vertex. The definition for a single-sink buy-at-bulk spanner where the terminal pairs share the same sink is defined similarly.

\subsection{Our Contributions}

To the best of our knowledge, we give the first efficient sublinear factor approximation algorithm for \bbs on $[\poly(n)]_{\pm}$ and $\R$. Our results even cover the case when the distances are negative.

Below we present our results and technical tools for solving buy-at-bulk spanners. Namely, \emph{distance-constrained junction trees} and \emph{resource-constrained shortest paths with negative consumption}.

\subsubsection{Buy-at-bulk spanners} \label{sec:intro-off}

We prove the following result for \ubbs on $[\poly(n)]_{\pm}$ in Section \ref{sec:spanner4by5}.

\begin{restatable}{theorem}{thmbbsfourfive} \label{thm:bbs45}
    For any constant $\ep > 0$, there exists a polynomial-time randomized algorithm for \ubbs on $[\poly(n)]_{\pm}$ with approximation ratio $\tO(n^{4/5 + \ep})$ and the distance constraints for all $(s,t) \in D$ are satisfied with high probability.\footnote{Throughout this paper, when we say high probability, we mean probability at least $1 - 1/n$.}
\end{restatable}

Theorem \ref{thm:bbs45} generalizes the $\tO(n^{4/5 + \ep})$-approximation algorithm for the unit-demand directed buy-at-bulk network design problem \cite{antonakopoulos2009} and directed weighted spanners \cite{gkl2023}, by allowing distance constraints.

Next, we consider \bbs on $\R$ by slightly relaxing the distance constraints. This allows us to obtain approximation algorithms for \bbs on $[\poly(n)]_{\pm}$ in terms of $k$ (Corollary \ref{cor:bbs}) later on. For notation convenience, we define some condition numbers. Let 
\begin{equation} \label{def:eta}
    \eta:= \frac{|\min\{\min_{e \in E}\{\ell_e\},0\}|}{\min_{(s,t) \in D} \{|\bdgt(s,t)|\}}
\end{equation}
and
\begin{equation} \label{def:xi}
    \xi:= \frac{\max_{(s,t) \in D} \{|\bdgt(s,t)|\}}{\min_{(s,t) \in D} \{|\bdgt(s,t)|\}}.
\end{equation}
Intuitively, $\eta$ denotes the ratio between the magnitude of the most negative edge length and the smallest absolute value of the budget.\footnote{$\eta$ cares about $\min_{e \in E}\{\ell_e\}$, but not about $\max_{e \in E}\{\ell_e\}$. This is because we can safely ignore edges that are much longer than the budget, but we cannot do so for edges (with negative lengths) that are much shorter than the budget.} If edge lengths are all non-negative, then $\eta=0$. Similarly, $\xi$ denotes the ratio between the largest and the smallest absolute value of the budget.
To accommodate a negative distance budget, we use the $\siign$ function
\begin{equation} \label{def:sign}
\siign(x) = \begin{cases} 
                 -1  & \text{if } x < 0, \\
                 0 & \text{if } x = 0, \\
                 1 & \text{if } x > 0.
            \end{cases}
\end{equation}
Suppose we are given a tolerance parameter $\theta > 0$. When the distance budget is positive, the goal is to satisfy the distance constraint within a factor of $(1+\theta)$. When the distance budget is negative, the distance between the terminal pair is below $(1-\theta)$ times the budget. 
We prove Theorem \ref{thm:bbs} in Section \ref{sec:spannerk}.

\begin{restatable}{theorem}{thmbbsr} \label{thm:bbs}
     For any $\theta > 0$ and constant $\ep > 0$, there are $\poly(n,1/\theta,\eta,\xi)$-time randomized algorithms for \bbs on $\R$ with approximation ratio $\tilde{O}(k^{1/2 + \ep})$ and for \ssbbs on $\R$ with approximation ratio $\tilde{O}(k^{\ep})$, both satisfying
     \begin{equation} \label{eq:theta-relaxed}
         \sum_{e \in p(s,t)} \ell_e \le (1+\theta \siign(\bdgt(s,t))) \bdgt(s,t)
     \end{equation}
     for each $(s,t) \in D$, with high probability. Here, $\eta$ and $\xi$ are the condition numbers defined in \eqref{def:eta} and \eqref{def:xi}, respectively. When $1/\theta, \eta, \xi  \in \poly(n)$\footnote{When $\xi$ is exponential in $n$, a workaround is to break all our demand pairs into $\log \xi$ buckets. Each bucket $i$ will have $\xi_i = 2$ but we will end up paying an extra $\log \xi$ cost (multiplicative).}, the algorithm runs in polynomial time.
\end{restatable}

We note that when the edge lengths are non-negative, $\eta = 0$, so the running time is $\poly(n,1/\theta,\xi)$.

From Theorem \ref{thm:bbs}, when the domain $\cD = [\poly(n)]_{\pm}$, we can set $\theta = 1 / \poly(n)$ (for a sufficiently large $\poly(n)$) such that the distance constraints are satisfied and the condition numbers $\eta$ and $\xi$ are polynomial in $n$ (since it is sufficient to consider that $\bdgt(s,t) \in \poly(n)$). This implies the following.

\begin{restatable}{corollary}{corbbs} \label{cor:bbs}
     For any constant $\ep > 0$, there are polynomial-time randomized algorithms for \bbs on $[\poly(n)]_{\pm}$ with approximation ratio $\tilde{O}(k^{1/2 + \ep})$ and for \ssbbs on $[\poly(n)]_{\pm}$ with approximation ratio $\tilde{O}(k^{\ep})$ satisfying the distance constraints with high probability.
\end{restatable}

Corollary \ref{cor:bbs} generalizes the $\tO(k^{1/2 + \ep})$-approximation algorithm for directed weighted spanners \cite{gkl2023} and buy-at-bulk network design \cite{antonakopoulos2009}, and the $\tO(k^{\ep})$-approximation algorithm for the single-source version of these problems, by allowing general demands, \ppc, and distance constraints.

We emphasize again that \bbs unifies the buy-at-bulk network design and spanners problems. Our results expand the domain for some results in the existing literature on these two problems. Furthermore, we allow edge lengths to be negative and of arbitrary magnitude (provided they are well-conditioned). There are no directed spanner results that account for edge lengths with negative values and arbitrary magnitude that we are aware of.

\subsubsection{Distance-constrained junction trees}

In previous literature, one main engine for solving the directed pairwise spanner problem \cite{chlamtavc2020approximating,grigorescu2021online}, the directed buy-at-bulk network design problem \cite{shen2020online,cekp,chekuri2010approximation,antonakopoulos2009}, and the directed Steiner forest problem \cite{berman2013approximation,charikar1999approximation,feldman2012improved,chekuri2011set} is the notion of \emph{junction tree}. The notion of junction tree  used for these problems is a union of an in-arborescence and an out-arborescence both rooted at the same vertex. For our purpose, we extend the definition of junction trees to handle both buy-at-bulk costs and distance constraints.

\begin{definition} \label{def:rcjt}
    Given an instance of \bbs and a \emph{root} vertex $r \in V$, a {\em distance-constrained junction tree} $\cJ:=\{p(s,r,t)\}$ rooted at $r$ is a collection of $s \leadsto r \leadsto t$ paths in $G$ that satisfy the distance constraint $\bdgt(s,t)$, for at least one $(s,t) \in D$. The $s \leadsto r$ paths form an in-arborescence and the $r \leadsto t$ paths form an out-arborescence, both rooted at $r$. Here, $p(s,r,t)$ denotes the set of edges in the $s \leadsto r \leadsto t$ path.
\end{definition}

We note that a junction tree $\cJ$ that connects a terminal pair $(s,t)$ within its budget $\bdgt(s,t)$ decides a unique $s \leadsto t$ path that passes through the root vertex $r$.\footnote{A junction tree does not necessarily have a tree structure in directed graphs. For directed pairwise spanners, directed buy-at-bulk network design, and directed Steiner forests, an edge in a junction tree may be used twice, once in the in-arborescence and once in the out-arborescence.}

Given a junction tree $\cJ$ rooted at $r$, the cost of the junction tree is defined as follows.
\begin{equation}
    \cost(\cJ) := \sum_{e \in \cup_{p(s,r,t) \in \cJ} p(s,r,t)} \sigma(e) + \sum_{e \in E} \delta(e) \sum_{p(s,r,t) \in \cJ: e \in p(s,r,t)} \dem(s,t).
\end{equation}

The crucial subproblem for solving \bbs on $[\poly(n)]_{\pm}$ is defined as follows.

\begin{restatable}{definition}{defjuncbbs} \label{def:junc_bbs_main}
    \rcjt 
    
    \textbf{Instance}: Same as \bbs on $[\poly(n)]_{\pm}$.
    
    \textbf{Objective}: Find a root $r \in V$, a resource-constrained junction tree $\cJ$ rooted at $r$, such that the ratio of the cost \eqref{obj:2m} of $\cJ$ to the number of $(s,t)$ pairs that satisfy the distance requirement 
    \[\res(p(s,r,t)) := \sum_{ e \in p(s,r,t)} \ell_e \le \bdgt(s,t)\]
    is minimized. Specifically, the goal is the following:
    \begin{equation}
             \min_{r \in V, \cJ}\frac{\cost(\cJ)}{|\{(s,t) \in D \mid \exists \: p(s,r,t) \text{ in } \cJ: \res(p(s,r,t)) \leq \bdgt(s,t)\}| }.
    \end{equation}
  
\end{restatable}

To further extend our results to the domain $\R$ for edge lengths, we use a relaxed version of distance-constrained junction trees. This allows us to obtain approximation algorithms for \rcjt (Corollary \ref{thm:mdrcjt}) later on.

\begin{definition} \label{def:rrcjt}
    Given an instance of \bbs and a constant $\theta > 0$, a {\em $\theta$-relaxed distance-constrained junction tree} $\cJ:=\{p(s,r,t)\}$ is a collection of $s \leadsto r \leadsto t$ paths that satisfy
    \[\sum_{ e \in p(s,r,t)} \ell_e \le (1+\theta \siign(\bdgt(s,t))) \bdgt(s,t)\] for at least one $(s,t) \in D$. The $s \leadsto r$ paths form an in-arborescence and the $r \leadsto t$ paths form an out-arborescence, both rooted at $r$.
\end{definition}
The crucial subproblem for solving \bbs on $\R$ is defined as follows.

\begin{restatable}{definition}{defthetajuncbbs} \label{def:theta_junc_bbs_main}
    \rrcjt 
    
    \textbf{Instance}: Same as \bbs on $\R$.
    
    \textbf{Objective}: Find a root $r \in V$, a resource-constrained junction tree $\cJ$ rooted at $r$, such that the ratio of the cost \eqref{obj:2m} of $\cJ$ to the number of $(s,t)$ pairs that approximately satisfy the distance requirement 
    \[\res(p(s,r,t)) := \sum_{ e \in p(s,r,t)} \ell_e \le (1+\theta \siign(\bdgt(s,t))) \bdgt(s,t)\]
    is minimized. Specifically, the goal is the following:
    \begin{equation}
             \min_{r \in V, \cJ}\frac{\cost(\cJ)}{|\{(s,t) \in D \mid \exists \: p(s,r,t) \text{ in } \cJ: \res(p(s,r,t)) \leq (1+\theta \siign(\bdgt(s,t))) \bdgt(s,t)\}| }.
    \end{equation}
\end{restatable}

We show the following in Section \ref{sec:jtree}.

\begin{restatable}{theorem}{thmmdrcjttheta} \label{thm:mdrcjttheta}
 For any constant $\ep > 0$, there is a $\poly(n,1/\theta,\eta,\xi)$-time randomized algorithm that gives an $\tO(k^\ep)$-approximation for \rrcjt with high probability. When $1/\theta,\eta,\xi \in \poly(n)$, the algorithm runs in polynomial time.
\end{restatable}

We note that when the edge lengths are non-negative, $\eta = 0$, so the running time is $\poly(n,1/\theta,\xi)$.

From Theorem \ref{thm:mdrcjttheta}, when the domain $\cD = [\poly(n)]_{\pm}$, we can set $\theta = 1 / \poly(n)$ (for a sufficiently large $\poly(n)$) such that the distance constraints are satisfied and the condition numbers $\eta$ and $\xi$ are polynomial in $n$ (since it is sufficient to consider that $\bdgt(s,t) \in \poly(n)$). This implies the following.

\begin{restatable}{corollary}{thmmdrcjt} \label{thm:mdrcjt}
 For any constant $\ep > 0$, there is a polynomial-time randomized algorithm that gives an $\tO(k^\ep)$-approximation for \rcjt with high probability.
\end{restatable}

Corollary \ref{thm:mdrcjt} generalizes the minimum density junction tree used for buy-at-bulk network design \cite{antonakopoulos2009}, pairwise spanners \cite{chlamtavc2020approximating}, and weighted spanners \cite{gkl2023}, with the same approximation ratio. We emphasize that our minimum-density junction tree framework accounts for buy-at-bulk costs and distance constraints with negative edge lengths.

\subsubsection{Resource-constrained shortest paths with negative consumption}

In previous literature, one crucial case analysis for solving directed spanner problems \cite{dinitz2011directed,berman2013approximation,chlamtavc2020approximating,grigorescu2021online,gkl2023} and the directed buy-at-bulk network design problem \cite{antonakopoulos2009} is via \emph{flow-based} linear programs (LP). The LP formulations for spanners potentially contain an exponential number of constraints and thus require an efficient black-box subroutine to be solved in polynomial time. A fully polynomial time approximation scheme (FPTAS) for the \emph{restricted shortest path problem} \cite{hassin1992approximation,lorenz2001simple} is treated as an approximate separation oracle to approximately solve the flow-based LP for spanners.

To further accommodate buy-at-bulk costs for spanners with negative edge lengths, we use an approximation scheme for the resource-constrained shortest paths problem with \emph{negative} resource consumption, a generalization of the FPTAS for the restricted shortest path problem \cite{horvath2018multi}.

\begin{restatable}{definition}{defnrcsp} \label{def:rcsp}

 \textsc{\rcsp ($m$-RCSP)}
    
  \textbf{Instance}: An $n$-vertex directed graph $G = (V,E)$, with edge costs $c: E \to \mathbb{R}_{\ge 0}$, a terminal pair $(s,t)$ with $s,t \in V$, and a resource budget $\boldsymbol{L} = (L_1, ..., L_m) \in (\R \setminus \{0\})^m$. For each edge $e \in E$, we have also have a resource consumption vector $\boldsymbol{w(e)} = (w_1(e),w_2(e),\ldots,w_m(e))$ of size $m$ where each $w_i(e) \in \mathbb{R} \: \forall i \in [m] $. Furthermore, we assume that there are no negative cycles induced by $\{w_i(e)\}_{e \in E}$ for each $i \in [m]$. 

\textbf{Objective}: Find a min-cost $s \leadsto t$ path $P$ such that $\sum_{e \in P} w_i(e) \leq L_i,  \: \forall i \in [m]$. The cost of $P$ is $\sum_{e \in P}c(e)$.
\end{restatable}

We note that when $m=1$ and the resource consumption is non-negative, $m$-RCSP captures the restricted shortest path problem. In the LP for \bbs, the constraints can be captured by solving $m$-RCSP. For $m$-RCSP, we design an algorithm that finds an optimal solution that slightly violates the resource budget. Let $\opt_{RCSP}$ be the cost of any minimum cost $s \leadsto t$ path that satisfies the resource constraints. Given a tolerance vector $(\eps_1, \eps_2 . . . , \eps_m) \in \R_{> 0}^m$, a $(1; 1 + \eps_1, 1+\eps_2 . . . , 1 + \eps_m)$-approximation scheme finds an $s \leadsto t$ path whose cost is at most $\opt_{RCSP}$, but the $i$-th resource constraint is satisfied up to a factor of $(1 + \eps_i)$ for that path. It is required that $L_i \cdot \eps_i > 0$. Let the condition number
\begin{equation}
    \gamma_i:=\frac{|\min\{\min_{e \in E}\{w_i(e)\},0\}|}{|L_i|}
\end{equation}
which denotes the ratio between the magnitude of the most negative $i$-th resource consumption among the edges and the absolute value of the budget for the $i$-th resource.

We show the following result in Section \ref{sec:rcsp}.

\begin{restatable}{theorem}{thmhorvathnegative} \label{thm:horvath_negative}
  There exists a $\poly(n^m,\gamma_1, ..., \gamma_m, 1/|\eps_1|, ..., 1/|\eps_m|)$-time $(1; 1 + \eps_1, 1+\eps_2 . . . , 1 + \eps_m)$-approximation scheme for $m$-RCSP. When $\gamma_i, 1/|\eps_i| \in \poly(n)$ $\forall i \in [m]$ and $m$ is a constant, the approximation scheme runs in polynomial time.  
\end{restatable}

\subsubsection{Summary}

We summarize our main results for \bbs in Table \ref{table:sum} by listing the approximation ratios and contrasting them with the corresponding known approximation ratios. The running time for \bbs and \ssbbs on $[\poly(n)]_{\pm}$ is polynomial. The running time for the $\theta$-feasible or relaxed results on $\R_{> 0}$ is $\poly(1/\theta,\eta,\xi,n)$.
\enote{in the table, can the results between col 2 and 3 be aligned propertly?}

\begin{table}[!htb]
\begin{center}
\def\arraystretch{1.2}
\begin{tabular}{|*3{l|}}
\hline

\textbf{Problem} & \textbf{Our Results} & \textbf{Previous Problems and Results} \\
\hline

\begin{tabular}{@{}l@{}} \textsc{Buy-at-bulk} \\ \textsc{Spanner} \\ on $[\poly(n)]_{\pm}$ \end{tabular} & \begin{tabular}{@{}l@{}} 
$\tO(n^{4/5+\ep})$ (unit-demand, Thm \ref{thm:bbs45}) \\ $\tO(k^{1/2+\ep})$ (Cor \ref{cor:bbs}) \\ $\tO(k^{\ep})$ (single-source, Cor \ref{cor:bbs}) 
\end{tabular} & \begin{tabular}{@{}l@{}} $O(n^{4/5+\ep})$ (unit-demand buy-at-bulk) \cite{antonakopoulos2009} \\ $\tO(n^{4/5 + \ep})$ (weighted spanners) \cite{gkl2023} \\ $O(k^{1/2+\ep})$ (buy-at-bulk) \cite{antonakopoulos2009} \\
$O(k^{\ep})$ (single-source buy-at-bulk) \cite{antonakopoulos2009} 
\\ $\tO(k^{1/2+\ep})$ (weighted spanners) \cite{gkl2023} \\ $\tO(k^{\ep})$ (single-source weighted spanners) \cite{gkl2023} \\ $\tO(n^{3/5 + \ep})$ (pairwise spanners) \cite{chlamtavc2020approximating} \end{tabular} \\
\hline
 \textsc{Buy-at-bulk} & $\tO(k^{1/2+\ep})$ ($\theta$-feasible, Thm \ref{thm:bbs}) & same as above, note that weighted spanners \\
 \textsc{Spanner} on $\R$ & $\tO(k^{\ep})$ ($\theta$-feasible, single-source, & consider edge lengths in $[\poly(n)]$ \\
 & Thm \ref{thm:bbs}) & \\
\hline

\textsc{Minimum-density} & $\tO(k^{\ep})$ (on $[\poly(n)]_{\pm}$,
Cor \ref{thm:mdrcjt}) & $O(k^{\ep})$ (buy-at-bulk) \cite{antonakopoulos2009} \\
\textsc{Distance-constrained} & $\tO(k^{\ep})$ (on $\R$, $\theta$-relaxed, Thm \ref{thm:mdrcjttheta}) & $\tO(k^{\ep})$ (pairwise spanners) \cite{chlamtavc2020approximating} \\
\textsc{Junction Tree} & & $\tO(k^{\ep})$ (weighted spanners) \cite{gkl2023} \\

\hline
$m$-RCSP with negative & an $(1; 1 + \eps_1, 1+\eps_2 . . . , 1 + \eps_m)$- & an $(1; 1 + \eps, 1+\eps . . . , 1 + \eps)$-\\
resource consumption & approximation scheme (Thm \ref{thm:horvath_negative}) & approximation scheme for $m$-RCSP \\
& & with non-negative consumption\\

\hline

\end{tabular}

\caption{Summary of the approximation ratios. Here, $n$ refers to the number of vertices and $k$ refers to the number of terminal pairs. The edge costs $\sigma$ and $\delta$ are non-negative rational numbers. $\theta$-feasible means that the resource constraints are satisfied within a factor of $(1 + \theta \siign(\bdgt(s,t)))$ for each $(s,t) \in D$. Input graphs do not have negative cycles induced by the edge lengths (or resources for $m$-RCSP).} \label{table:sum}
\end{center}
\end{table}

\subsection{High-level technical overview}

\subsubsection{The $\tO(n^{4/5 + \eps})$-approximation algorithm for \bbs on $[\poly(n)]_{\pm}$}

Recall that we are given a directed graph with terminal vertex pairs. Each terminal vertex pair has a demand and a distance budget. Each edge is associated with a length that can be negative and a buy-at-bulk cost. There are no negative cycles induced by the edge lengths. The goal is to output a collection of routes with minimum cost that satisfy the terminal demands and distance constraints.

Similar to the Steiner forest framework \cite{berman2013approximation,feldman2012improved} and the buy-at-bulk network design framework \cite{antonakopoulos2009}, we classify the terminal pairs as follows:
\begin{itemize}
    \item A pair $(s,t) \in D$ is \emph{good} if the feasible (i.e., satisfying the distance constraint) and cheap (low \upc and low \ppc) $s \leadsto t$ paths span a great number of vertices in $V$.
    \item A pair $(s,t) \in D$ is \emph{bad} otherwise.
\end{itemize}

With distance constraints, we have to handle \emph{all} the cases carefully, without destroying the desired structural properties. The analysis significantly departs from \cite{feldman2012improved,berman2013approximation,gkl2023,antonakopoulos2009} in several aspects, as we describe below.

\paragraph{Handling good pairs.} For the first case, a standard approach is to sample a sufficient number of intermediary vertices from $V$ and add cheap and feasible paths to connect the good pairs. For directed Steiner forests, edges only have \upc and there are no distance constraints, so it is sufficient to add cheap paths \cite{feldman2012improved,berman2013approximation}. For weighted spanners \cite{gkl2023}, edges only have \upc and positive edge lengths, so it is sufficient to find a \emph{shortest cheap path} by using the restricted shortest path FPTAS \cite{lorenz2001simple,hassin1992approximation} as a subroutine. 

For \bbs on $[\poly(n)]$, edges have positive edge lengths, \upc, and \ppc, one can carefully use the resource-constrained shortest path (RCSP) FPTAS from \cite{horvath2018multi} as the subroutine to handle distance constraints. This approach connects the sources to the intermediary vertices and the intermediary vertices to the sinks, where the connecting paths have short distances and \emph{low \upc and \ppc}. Handling negative edge lengths requires a more general subroutine to connect the terminal vertices to (from) the intermediary vertices, so we modify the FPTAS for RCSP to accommodate negative edge lengths (described in more detail in Section \ref{sec:intro-horvath}).

\paragraph{Handling bad pairs.}
After the good pairs are resolved, we iteratively use a greedy algorithm based on a density argument. The greedy algorithm constructs a low-density partial solution in each iteration. Adding low-density partial solutions iteratively guarantees a global solution of approximately minimum cost. For ease of our analysis, we partition the bad pairs into three classes based on an unknown optimal solution $\cP^* = \{p^*(s,t)\}_{(s,t) \in D}$.

\begin{itemize}
    \item A bad pair $(s,t) \in D$ is in \emph{class 1} if $p^*(s,t)$ has a high \ppc.
    \item A bad pair $(s,t) \in D$ is in \emph{class 2} if $p^*(s,t)$ has a high \upc and a low \ppc.
    \item A bad pair $(s,t) \in D$ is in \emph{class 3} if $p^*(s,t)$ has a low \upc and a low \ppc.
\end{itemize}
We consider three subcases for the bad pairs.
\begin{itemize}
    \item Case 1: the number of bad pairs is small.
    \item Case 2: the number of class 2 bad pairs is larger than the number of class 3 bad pairs.
    \item Case 3: the number of class 2 bad pairs is at most the number of class 3 bad pairs.
\end{itemize}

The iterative procedure continues by picking the better solution by trying out case 2 and case 3 until the number of bad pairs is small (case 1). Once we reach case 1, we use Corollary \ref{cor:bbs} (obtained from Theorem \ref{thm:bbs}) to resolve all the bad pairs. Since the number of unresolved vertex pairs is small in this case, the approximation ratio is sufficiently small.

A key observation is that the number of class 1 bad pairs is always smaller than the threshold used for case 1. Therefore, it suffices to consider cases 2 and 3 before we run out of class 2 and class 3 bad pairs.

For case 2, there are more class 2 pairs than class 3 pairs. Since the class 2 bad pairs have high \upc, we must have an edge that belongs to a sufficient number of paths that connect the terminal pairs within the required distance in an optimal solution. This implies that there must be a vertex that lies on an edge used plenty of times, so we can use Corollary \ref{thm:mdrcjt} (obtained from Theorem \ref{thm:mdrcjttheta}) to find a low-density junction tree rooted at that vertex as our partial solution.

Recall that for each bad pair, the number of vertices spanned by its feasible paths is small. In case 3, there are more class 3 pairs than class 2 pairs. Therefore, there are sufficient terminal pairs whose \upc and \ppc are low and whose feasible paths span a small number of vertices. A natural approach that addresses this case is via a flow-based LP formulation. Intuitively, when we have a large number of terminal pairs with low \upc and \ppc and a small number of vertices spanned by feasible paths, the edge indicators of the LP must be large enough to construct a feasible fractional solution. To solve the LP, we use our RCSP framework that accommodates negative resource consumption as a separation oracle to extract violating constraints. After using a careful pruning procedure as in \cite{antonakopoulos2009} and a rounding scheme similar to the ones used for Steiner forests \cite{berman2013approximation} and weighted spanners \cite{gkl2023}, we extract the edges with high indicator values in the LP and recover the collection of paths as our partial solution.

\subsubsection{Approximations for \bbs in terms of $k$}
For \ssbbs, the result directly follows by Theorem \ref{thm:mdrcjttheta}. The proof for \bbs follows a standard iterative density procedure for directed Steiner forests \cite{chekuri2011set} and spanners \cite{grigorescu2021online,gkl2023}. We show that iteratively picking minimum-density distance-constrained junction trees only pays a factor of $O(\sqrt{k})$. Combining this and Theorem \ref{thm:mdrcjttheta} results an $\tO(k^{1/2 + \ep})$-approximation.

\subsubsection{The $\tO(k^{\eps})$-approximation for \rcjt}

At a high level, junction trees form a cheap partial solution that connects a subset of terminal pairs. For \bbs, a potential approach is to iteratively select a low-density junction tree in case there exist edges that are crucial for connecting terminal pairs. For our purpose, we have to construct \emph{feasible} junction trees that satisfy the resource constraints while connecting terminal pairs. The modified junction-tree-based approach is the main engine of our framework.

Suppose we have a fixed root vertex $r \in V$, and the goal is to find a minimum-density junction tree rooted at $r$. Here, the \emph{density} is defined as the cost of the junction tree divided by the number of terminal pairs connected. The framework in \cite{chlamtavc2020approximating} used for pairwise spanners with unit lengths has three main steps: 1) construct a layered graph from $G$ to capture the distance constraints and a junction tree rooted at $r$, 2) use the \emph{height reduction} technique to construct a tree-like graph from the layered graph by paying a small approximation ratio, and 3) use a linear programming (LP) formulation on the tree-like graph and round the fractional solution. The main reason for the second and third steps is that the tree-like graph is well-structured, which allows one to formulate an LP with a polylogarithmic integrality gap.

To capture distance constraints with negative edge lengths and the flow-based cost in the buy-at-bulk problem, our implementation requires several new ideas, which significantly depart from the analysis of \cite{chlamtavc2020approximating} in several aspects, as we describe below.

\paragraph{Scaling and rounding the edge lengths.} Before the first step, we use an approach similar to \cite{horvath2018multi} to properly scale and round the edge lengths. The edge lengths are rounded up so that the distances between the terminal pairs might be slightly overestimated. This allows us to construct a layered graph of size polynomial in the condition numbers which approximately preserves the edge lengths and terminal distances.

\paragraph{Turning distance constraints into connectivity constraints.}
In the first step, we construct a layered graph that approximately preserves the edge lengths. 
In the layered graph, each layer captures the distance to (from) the root vertex. 
To handle negative edge lengths, a key modification is to allow edges to go \emph{backward}, instead of always forward as in \cite{chlamtavc2020approximating}. Furthermore, since we have general edge lengths instead of unit edge lengths, it is no longer the case that only neighboring layers have edges between them. Edges are added from one layer to another whenever their length corresponds to the distance between layers.

\paragraph{Handling \ppc.}
In the second step, the main challenge is to handle the flow-based cost properly. Fortunately, following the height reduction technique in \cite{cekp} allows us to reduce the two-metric problem to a variant of the Steiner problem, namely, \emph{minimum density Steiner label cover}, which only accounts for the \upc. In this problem, the goal is to find a minimum density subgraph that connects \emph{pairs of terminal vertex sets}, subject to a relation induced by the distance constraints.

\paragraph{LP formulation and rounding.}
In the third step, we use an LP formulation for minimum density Steiner label cover \cite{chlamtavc2020approximating}. 
The rounding approach extracts a \emph{cross-product subset} of the terminal pair sets, thus allowing one to use a standard rounding scheme for the group Steiner problem on trees. Ultimately, we extract a cross-product subset and use the group Steiner rounding scheme for our distance-constrained problem.

\subsubsection{The approximation scheme for RCSP with negative consumption} \label{sec:intro-horvath}

Recall that for $m$-RCSP, we are given a directed graph with a terminal vertex pair. Each edge is associated with a cost and an $m$-dimensional resource consumption vector where entries can be negative. There is no negative cycle induced by any resource type. We also have an $m$-dimensional budget for the resource consumption. The goal is to find a cheap path to connect the terminal pair without exceeding the budget.

To construct the approximation scheme, we follow the dynamic programming (DP) paradigm in \cite{horvath2018multi}. To approximately preserve feasibility, the approach in \cite{horvath2018multi} scales the resource consumption properly, and the DP memorizes the feasible resource consumption \emph{patterns} while reaching a specific vertex from the source. To accommodate negative resource consumption, the main challenge is the negative resource consumption can be \emph{unbounded} in terms of the scale for the non-negative consumption. Addressing these challenges requires several modifications. First, we construct a larger DP table where the number of resource consumption patterns depends on the condition numbers $\gamma_i$. Second, with the assumption that negative cycles do not exist, our DP also considers \emph{hop counts} in a fashion similar to the Bellman-Ford algorithm.

\subsection{Related work}

\paragraph{Directed spanners.} 
A well-studied variant of spanners is called the \emph{directed $s$-spanner} problem, where 
there is a fixed value $s \geq 1$ called the \emph{stretch}, and the goal is 
to find a subgraph with a minimum number of edges such that the distance between \emph{every} pair of
vertices is preserved  within a factor of $s$ in the original
graph. When the lengths of the edges are uniform and $s=2$, there is a tight $\Theta(\log n)$-approximation algorithm \cite{elkin1999client, Kortsarz2001OnTH}.  When $s=3,4$ there are $\tO(n^{1/3})$-approximation algorithms \cite{berman2013approximation,dinitz2016approximating}. When $s > 4$, the best known approximation factor is $\tO(n^{1/2})$ \cite{berman2013approximation}. The problem is hard to approximate within an $O(2^{{\log^{1-\eps} n}})$ factor for $3 \leq s = O(n^{1-\delta})$ and any $\eps, \delta \in (0,1)$, unless $NP\subseteq  DTIME(n^{\operatorname{polylog} n})$ \cite{ElkinP07}. More general variants consider the \emph{pairwise
  spanner} problem \cite{chlamtavc2020approximating}, and the {\em client-server} model \cite{elkin1999client,bhattacharyya2012transitive}, where the set of
terminals  is arbitrary $D = \{(s_i,t_i) \mid i \in [k]\} \subseteq V \times V$, and each
terminal pair $(s_i,t_i)$ has its own target distance $d_i$. The goal is
to compute a minimum cardinality subgraph in which for each $i$, the
distance from $s_i$ to $t_i$ is at most $d_i$.  For the pairwise
  spanner problem with uniform lengths, \cite{chlamtavc2020approximating} obtains an $\tO(n^{3/5 + \ep})$
approximation.
Recently, \cite{gkl2023} studied the {\em weighted} spanner problem for arbitrary terminal pairs, which has a closer formulation to buy-at-bulk spanners. \cite{grigorescu2021online} studied online directed spanners. We refer the reader to 
the excellent survey \cite{ahmed2020graph} for a more comprehensive exposition.

\paragraph{Buy-at-bulk network design.} The buy-at-bulk problem has received considerable attention and has been well-studied in the past few decades.
Most of the previous buy-at-bulk literature focused on undirected networks, as listed below.

The problem was first introduced in \cite{salman2001approximating}, where the subadditive load function was used to capture the economy of scale in network design.
\cite{salman2001approximating} showed that the problem is $NP$-hard, and gave an $O(\min\{\log n, \log D_{\max}\})$-approximation algorithm for the single-source problem, where $D_{\max}$ is the maximum demand.
When the edge costs are uniform, there is a $\polylog(n)$-approximation algorithm for the multi-commodity problem \cite{awerbuch1997buy} and $O(1)$-approximation algorithms for the single-source problem \cite{guha2009constant,talwar2002single,gupta2003simpler}.
With non-uniform load functions, the first nontrivial result for the multi-commodity problem was $O(\log D_{\max}\exp(O(\sqrt{\log n \log \log n})))$-approximate \cite{charikar2005non}, later improved to $\polylog(n)$-approximate \cite{chekuri2010approximation}; for the single-source problem, there is an $O(\log k)$-approximate algorithm \cite{meyerson2008cost}. 
The buy-at-bulk problem is more intractable on directed graphs.
The state-of-the-art is an $\min\{\tO(k^{1/2+\ep}),\tO(n^{4/5+\ep})\}$-approximate algorithm \cite{antonakopoulos2009}.
Even in the special case of directed Steiner forests, previous algorithms only gave $\poly(n)$ but sublinear (i.e., $o(n)$) approximation algorithms \cite{berman2013approximation,chlamtavc2020approximating,feldman2012improved,abboud2018reachability}.
On the hardness side, for the undirected buy-at-bulk problem, the multi-commodity problem is $O(\log^{1/2 - \ep} n)$-hard to approximate when the costs are general and $O(\log^{1/4 - \ep} n)$-hard to approximate when the costs are uniform \cite{andrews2004hardness}, and the single-source problem is $O(\log \log n)$-hard to approximate \cite{chuzhoy2008approximability}.
These results are based on the assumption that $NP \nsubseteq ZTIME(n^{\polylog(n)})$.
For directed buy-at-bulk, the problem is hard to approximate within an $O(2^{\log^{1-\ep} n})$ factor assuming that $NP \nsubseteq DTIME( n^{\polylog(n)})$, even in the special case of directed Steiner forests \cite{dodis1999design}.

\paragraph{Other variants of buy-at-bulk network design.} 
Besides the edge-weighted buy-at-bulk problem, there are other variants including the node-weighted problem and the prize-collecting problem. In the prize-collecting problem, each terminal pair also has a penalty and one can choose not to connect the pair and incur the penalty in the cost. Most of these results are on undirected graphs. For undirected node-weighted buy-at-bulk, there exists a polylogarithmic polynomial-time approximation algorithm \cite{chekuri2007approximation} and a polylogarithmic quasi-polynomial-time competitive online algorithm \cite{cekp}. In a more restricted case, i.e., undirected node-weighted Steiner trees, a polylogarithmic approximation algorithm was presented in \cite{klein1995nearly} and a polylogarithmic competitive online algorithm was presented in \cite{naor2011online}. Following \cite{naor2011online}, \cite{hajiaghayi2013online} extends to online undirected node-weighted Steiner forests while \cite{hajiaghayi2014near} extends to the online prize-collecting versions. These results fall under the unifying framework of \cite{aaabn-set-cover} which utilizes an online primal-dual LP rounding scheme. The work \cite{cekp} further extends to the online price-collecting buy-at-bulk problem with the same competitive ratio as the standard edge-weighted problem on both directed and undirected graphs. The work \cite{gupta2017last} considers a metric-based variant of undirected online buy-at-bulk and presents a framework that finds a cheap subgraph (compared to the minimum spanning tree or the optimal Steiner forest) that connects the terminal pairs with low stretch.

\paragraph{Resource-constrained shortest path.} The resource-constrained shortest path problem was introduced in \cite{horvath2018multi}. The input consists of $m$ resource types and a directed graph where each edge is associated with a non-negative cost and a non-negative (for all coordinates) resource consumption vector. The goal is to find a minimum-cost path that connects the single source vertex to the single sink vertex while satisfying the resource constraint. When $m=1$, this problem is equivalent to the restricted shortest path problem \cite{hassin1992approximation,lorenz2001simple}, which has been extensively used in the literature of spanners \cite{dinitz2011directed,chlamtavc2020approximating,feldman2012improved,berman2013approximation,grigorescu2021online}. The results of \cite{horvath2018multi} show that when $m$ is a constant, there exists an FPTAS that finds a path with a cost at most the same as the feasible minimum-cost path by violating each budget by a factor of $1+\ep$. The FPTAS for $m=2$ is used to solve the weighted spanner problem \cite{gkl2023}.

\subsection{Organization}
We present the $\tO(n^{4/5 + \eps})$-approximation algorithm for \bbs on $[\poly(n)]_{\pm}$ in Section \ref{sec:spanner4by5}. We present the $\tO(k^{1/2 + \ep})$-approximation algorithm for \bbs on $\R$ and the $\tO(k^{\ep})$-approximation algorithm for \ssbbs on $\R$ that may slightly violate the distance constraints in Section \ref{sec:spannerk}. We present the $\tO(k^\ep)$-approximation algorithm for \rrcjt in Section \ref{sec:jtree}. We present our RCSP framework in Section \ref{sec:rcsp}.

\section{An $\tO(n^{4/5 + \eps})$-approximation for Buy-at-Bulk Spanners} \label{sec:spanner4by5}

Recall Definition \ref{def:bbs}.

\defbbs*

Recall Theorem \ref{thm:bbs45}

\thmbbsfourfive*

Throughout this subsection, we set $\cD = [\poly(n)]_{\pm}$.  In this section, we prove Theorem \ref{thm:bbs45}.

\nnote{remove the word resource if we have time, otherwise the above line makes it work.}

As in \cite{feldman2012improved,berman2013approximation,gkl2023,antonakopoulos2009} let $\tau$ be our guess of the optimal solution - $\opt$ such that $\opt \leq  \tau \leq 2\cdot \opt$. Throughout this section, we will assume that our demand is uniform (i.e., $\dem(s,t) = 1 \: \forall \: (s,t) \in D$). If we have non-uniform demand, we can change it into uniform demand using a simple reduction by a standard reduction where we break our overall instances into $O(log \max_{(s,t)\in D} dem(s,t))$ smaller instances. Each of these smaller instances have uniform demand (this is the same as splitting a positive integer into powers of $2$).

We define some common notation used in the literature for Spanners, Steiner forests, and Buy at Bulk network design. Let $\beta = n^{3/5}$ and $L_1 = \tau/n^{4/5}, L_2 = n^{4/5} \tau /k$ - we will use these parameters later in our algorithm. We say that any path $p(s,t)$ connecting a terminal pair $(s,t)$ is  {\em feasible} if $\sum_{e \in p_{s,t}} \ell_e \leq \bdgt(s,t)$. 

We say that a path $p_{s,t}$ has cheap investment if $\sigma(p_{s,t}) =  \sum_{e \in p_{s,t}} \sigma(e) \leq L_1$; we say that a path has cheap maintenance if $\delta(p_{s,t}) =  \sum_{e \in p_{s,t}} \delta(e) \leq L_2$. Further, we say that $p(s,t)$ is {\em cheap} if it has both cheap investment and cheap maintenance \footnote{Recall that we only deal with uniform demand because we can reduce general demand to uniform demand with a cost of $log \sum_{(s,t)\in D} dem(s,t)$ - thus we don't need to multiply $\delta(e)$ with demand for a single path}. 

We call a terminal pair $(s,t) \in D$ {\em good} if the {\em local graph} $G^{s,t} = (V^{s,t}, E^{s,t})$ induced by the vertices on feasible $s \leadsto t$ paths that are cheap has at least $n/\beta$ vertices; we say it is {\em bad} otherwise. Let $D_a$ be the set of good pairs and $D_b$ be the set of bad pairs; also let $k_a = |D_a|$ and $k_b = |D_b|$. Our definition of good and bad pairs here has to account for both negative lengths and the addition of the \ppc unlike previous literature \cite{berman2013approximation,feldman2012improved,gkl2023}. 
Finally, we state that a set of paths $\{p(s,t)\}$ resolves(or settles) a pair $(s,t) \in D$ if it contains a feasible $s \leadsto t$ path.

\subsection{Resolving good pairs} \label{se:thick_negative}
We first define, $S = \{s \mid \exists t: (s,t) \in D\}$ and $T = \{t \mid \exists s: (s,t) \in D\}$. In this subsection, we settle the good pairs with high probability. We do this by sampling some vertices using  Algorithm \ref{alg:sampling_thick} and then adding some incoming paths and outgoing paths from the samples to the vertices in $S$ and $T$ respectively using Algorithm \ref{alg:resolve_thick}. We ensure that any path we build is both feasible and cheap. 

\begin{algorithm}[!htb]
\caption{Sample$(G(V,E))$} \label{alg:sampling_thick}
\begin{algorithmic}[1]
\State{$R \gets \phi$, $k \gets 3 \beta \ln n$.}

\State Sample $k$ vertices independently and uniformly at random and store them in the set $R$.

\State \Return $R$.
\end{algorithmic}
\end{algorithm}

\begin{restatable}[]{claim}{samph}
\label{cl:sampling_thick}
Algorithm \ref{alg:sampling_thick} selects a set of samples $R$ such that with high probability any given good pair $(s,t)$ has at least one vertex from its local graph in $R$.
\end{restatable}
\begin{proof}
    This standard claim and close versions of it have been proved in several articles (see for instance Claim 2.1 in \cite{gkl2023}).
\end{proof}
In Algorithm \ref{alg:resolve_thick}, we call Algorithm \ref{alg:sampling_thick} to get a set of samples $R$. For each $u \in R,s \in S,t \in T$, we try to add a set of $s \leadsto u$ paths and a set of $u \leadsto t$ path each of cost both \upc at most $O(L_1)$ and \ppc at most $L_2$. It is possible we can't find some of these paths and if that happens, we just ignore this pair and continue.

Now, we just need a black box algorithm that can add a $s \leadsto u$ path that fits our requirements (i.e., allowing negative edge lengths, multiple length constraints for the same $s \leadsto u$ pair). We use Algorithm \ref{alg:horvath_modified} for this purpose. Recall the \rcsp problem. 

\defnrcsp*

We also present a slightly modified version of Theorem \ref{thm:horvath_negative} in Corollary \ref{cor:rcsprational} that can exactly satisfy $m-1$ resource constraints where the corresponding resources are integers polynomial in $n$; and approximately satisfy one resource constraint that where the corresponding resource is a non-negative rational number. 

\begin{restatable}{corollary}{corrcsprational} 
\label{cor:rcsprational}
    \label{cl:horvath_modified_onerational}
    When for all edges $e \in E$, $w_{e,i} \in [\poly(n)]_{\pm} \:\: \forall i \in \{1,2,\ldots,m-1\}$ and $w_{e,m} \in Q_{\geq 0}$, there exists a fully polynomial time $(1;1,1,\ldots,1+\zeta)-$ algorithm for the $k-$\rcsp problem that runs in time polynomial in input size and $1/\zeta$.
\end{restatable}

\begin{proof}
    See Section \ref{sec:rcsp}.
\end{proof}

Since we have to ensure that both \upc and \ppc are low enough, we need to model one of them as a constraint and the other as an objective. Without loss of generality, we set the \upc as the objective in Corollary \ref{cl:horvath_modified_onerational} and set the \ppc as the $2 ^{nd}$ constraint (which is allowed to be rational and non-negative). The length is modeled as the first constraint (it is an integer polynomial in $n$). 

Now, we do not know the exact length we need for a $s \leadsto u$ path. But, what we can do is search for all possible lengths in the interval:$[n \cdot \lenlow, n \cdot \lenup]$\footnote{This can be speed up by binary search. We use linear search to improve readability.} for the lowest possible length as in \cite{gkl2023}. Using Claim \ref{cl:horvath_modified_onerational} as our black box, we get a {\em cheap} and feasible path $p(s,t)$ if such a path exists. 

\begin{algorithm}[!htb]
\caption{Thick pairs resolver $(G(V,E),\{\ell(e),\sigma(e),\delta(e)\}_{e\in E})$} \label{alg:resolve_thick}
\begin{algorithmic}[1]
\State{$R \gets \phi$, $P' \gets \phi$.}

\State{$R \gets \text{Sample}(G(V,E))$.}

\For{$u \in R$}
\For{$s \in S$}
\For{$len \in [n \cdot \lenlow, n \cdot \lenup]$}
\State Use Claim \ref{cl:horvath_modified_onerational} to find the cheapest $s \leadsto u$ path (in terms of \upc) that has \ppc $\leq L_2$ and $\len(p) \leq len$. If a path is found and it has \upc $\leq L_1$, add it to $P'$ and go to the next pair of terminals.

\EndFor
\State If no path can be found for any length, continue to the next iteration.
\EndFor
\EndFor

\For{$u \in R$}
\For{$t \in T$}
\For{$len \in [n \cdot \lenlow, n \cdot \lenup]$}
\State Use Claim \ref{cl:horvath_modified_onerational} to find the cheapest $u \leadsto t$ path (in terms of \upc) that has \ppc $\leq L_2$ and $\len(p) \leq len$. If a path is found and it has \upc $\leq L_1$, add it to $P'$ and go to the next pair of terminals.

\EndFor
\State If no path can be found for any length, continue to the next iteration with the next pair of terminals.
\EndFor
\EndFor

\For{$(s,t) \in D$}
\State Find some feasible $s \leadsto t$ path $p_3 $ that is obtained by joining a $s \leadsto u$ path $p_1 \in P'$ and $u \leadsto t$ path $p_2 \in P'$. This path will have \ppc at most $2(1+\zeta)L_2$, \upc at most $2 L_1$ (if no such path can be formed from $P'$, then ignore this pair). 

\State Add $p_3$ to $\mathcal{P}_g$.
\EndFor

\State \Return $\mathcal{P}_g$
\end{algorithmic}
\end{algorithm}

\begin{lemma} \label{le:thick_final}
    With high probability, the set of paths $P_g$ returned by Algorithm \ref{alg:resolve_thick} resolves all good pairs in $D$ with a total cost $\tO(n^{4/5} \cdot \tau)$. Moreover, Algorithm \ref{alg:resolve_thick} runs in polynomial time.
\end{lemma}

\begin{proof}
   
    If some $u \in R$ was originally in the local graph $G^{s,t}$, then Algorithm \ref{alg:resolve_thick} would have added at least one $s \leadsto u \leadsto t$ path from $G^{s,t}$ that is {\em feasible}. This is because if $u$ was in the local graph of $(s,t)$, then there exists an $s \leadsto u \leadsto t$ path $p$ of \upc less than $L_1$, \ppc less than $L_2$. This path $p(s,u,t)$ has a length   $\res(p(s,u,t)) \leq \bdgt(s,t)$. Let $p(s,u,t)$ be composed of two paths: a $s \leadsto u$ path $p(s,u)$ and a $u \leadsto t$ path $p(u,t)$. Both $p(s,u)$ and $p(u,t)$ will have \upc $\leq L_1$ and \ppc $\leq L_2$. 

    Algorithm \ref{alg:resolve_thick} will add the shortest $s \leadsto u$ path $p_1$ that has \upc $\leq L_1$ and \ppc $\leq L_2$ and this path will have length at most $\res(p(s,u))$. Similarly, Algorithm \ref{alg:resolve_thick} will also add the shortest $u \leadsto t$ path $p_2$ that has \upc $\leq L_1$ and \ppc $\leq L_2$ and this path will have length at most $\res(p(u,t))$. When we combine $p_1$ and $p_2$, we will get a path $p_3$ such that $\res(p_3) = \res(p_1) + \res(p_2) \leq \res(p(u,t)) + \res(p(s,u)) = \res(p(s,u,t)) \leq \bdgt(s,t)$. Further, $p_3$ will have an \upc $\leq 2 L_1$ and \ppc $\leq 2(1+\zeta) L_2$. Fixing $\zeta = 1$, the \ppc of $p_3$ would be $\leq 4 L_2$.
   
    Now we analyze the overall cost of Algorithm \ref{alg:resolve_thick}. The total cost of this procedure would be $O(|R| \cdot L_1  \cdot n  + k\cdot L_2)$. This is because we add one path for every sample from and to every vertex in $S$ and $T$ respectively, and each of these paths is cheaper than $O(L_1)$ to set up initially (we don't add a path otherwise). Further, we only have $k$ units of demand in total (recall that we have uniform demand), so their usage cost is going to be less than $O(L_2\cdot k)$. 

    Plugging in the values for $|R|$, $L_1$ and $L_2$, we can see that the total cost would be $\tO(n^{4/5} \cdot \tau)$. 
\end{proof}

\subsection{Resolving bad pairs} \label{subsec:pws-thin}

For this section, we adapt and expand a proof in \cite{antonakopoulos2009}. We also add some detail for some parts of the proof. We also need an entirely different proof technique (based on \cite{gkl2023}) for some parts of our proof. Before running the following algorithm, we first run the algorithm for good pairs and remove every resolved pair.

Let $P^* = \{P*_{uv}\}$ be an optimal solution (note that since it is an optimal solution, all paths here are feasible). Now let, 
\begin{itemize}
    \item $D_b^1 = \{(u,v)\in D_b \mid \delta(P^*_{uv}) > L_2\}$,
    \item $D_b^2 = \{(u,v)\in D_b \setminus D_b^1 \mid \sigma(P^*_{uv}) >  L_1\}$, and
    \item $D_b^3 = D_b \setminus (D_b^1 \cup D_b^2)$,
\end{itemize}

If $k_b \leq 4 n^{6/5}$, then we can use Corollary \ref{cor:bbs} to resolve all the bad pairs with cost $\leq \tO(n^{3/5+\ep})$ for some $\eps > 0$. 

Thus, we can assume that $k_b > 4 n^{6/5}$. Now, since each $(u,v) \in D_b^1$ has cost at least $L_2 = n^{4/5} \tau /k$, the cost of an optimal solution is $\tau$, and $k< n^2$, we have, $|D_b^1| < 2k /n^{4/5} < 2n^{6/5}$. This also implies that if we ever have a situation where all pairs in $D_b^2$ and $D_b^3$ are resolved, then as $k_b > 4 n^{6/5}$, $|D_b^1| < k_b/2$. Now, we have two cases based on whether $|D_b^2| > |D_b^3|$ or not. 

We define the {\em density} of a set of paths $P$ to be the ratio of the total cost of these paths to the number of pairs settled by those paths. Note that the total cost of a single path is the sum of the \upc and \ppc (since we have uniform demand). When we take sets of paths, we count every edge only once for \upc.

We first see how to efficiently construct a subset $P_1$ of paths with density $\tO(n^{4/5+\eps})\tau/|D|$. Then
we iteratively find paths with that density, remove the pairs corresponding to those paths, and repeat until we resolve all bad pairs. This gives a total cost of $\tO(n^{4/5+\eps})\tau$. We construct $P_1$ by building two other sets $P_2$ and $P_3$ and picking the smaller density of them.

Since $|D_b^1| < k_b/2$, 
when $|D_b^2|> |D_b^3|$, we have $|D_b^2| > k_b/4$.  Similarly, when $|D_b^2| < |D_b^3|$, we have $|D_b^3| > k_b/4$. 

\subsubsection{When $|D_b^2| > |D_b^3|$} \label{se:thin_junc}

We will use \rcjt as a black box for resolving this case. 

The following lemma is a slight variant of a standard result in multiple publications(\cite{berman2013approximation,feldman2012improved,gkl2023}).

\begin{claim} \label{cl:junction_tree_helper}
    If $|D_b^2|/2 \geq |D_b^3|$ (and thus $|D_b^2| \geq k_b/4$), then there exists a \rcjt of density $ O\left( n^{4/5} \cdot \tau/k_b\right)$.
\end{claim}
\begin{proof}

Observe that $\sum_{(u,v) \in D_b^2} \sigma(P_{uv}^*) > |D_b^2| \cdot \tau / n^{4/5}$. Now by a standard counting argument \cite{antonakopoulos2009,berman2013approximation,feldman2012improved,gkl2023} there must be an edge that belongs to at least $|D_b^2| /n^{4/5}$ of the paths in $P^*$. Now, consider a junction tree rooted at one of the vertices of this edge, and consists of all the paths going through this edge in $P^*$. This junction tree has cost $\leq \tau$ (since its a subgraph of an optimal solution $P^*$)  and can resolve at least $|D_b^2| /n^{4/5}$ pairs. Therefore, since $|D_b^2|/2 \geq k_b/4$, this junction tree will have a density $ \leq O\left( n^{4/5} \cdot \tau/k_b\right)$.
\end{proof}

\begin{lemma} \label{le:thin_costly_final}
    When $|D_b^2| > |D_b^3|$, we can get a set of paths $P_2$ that has density at most $\tO(n^{4/5+\ep} \cdot \tau/k_b)$    
\end{lemma}
\begin{proof}
    From Claim \ref{cl:junction_tree_helper}, there exists a \rcjt of density at most $O(n^{4/5} \cdot \tau/k_b)$. We use Corollary \ref{thm:mdrcjt} to get a \rcjt with density at most $\tO(n^{4/5+\ep} \cdot \tau/k_b)$ and store the paths returned by it in $P_2$.
\end{proof}

\subsubsection{When $|D_b^3| \geq |D_b^2|$} \label{se:thin_lp}

To handle this case, we first build a linear program and solve it. Unfortunately, the solution to the LP does not immediately fulfill our exact needs (although it is close). So, we do some careful processing to turn the solution to fit our exact requirements. This section is based on \cite{antonakopoulos2009,gkl2023}. 

\begin{subequations} \label{lp:thin_pair_original}
\begin{align}
& \min & & \sum_{e \in E}{\sigma(e) \cdot x_e} \label{lp:thin_pair_original2a} \\
& \text{subject to}
& & \sum_{(u,v) \in D_b} y_{uv} \geq k_b/4, \label{lp:thin_poly_1}\\
& & &\sum_{\Pi \ni p \ni e} f_p \leq x_e & \forall (u,v) \in D_b, e \in E,\label{lp:poly_2}\\
& & &\sum_{p \in \Pi(u,v)} f_p \geq y_{u,v} & \forall (u,v) \in D_b, \label{lp:poly_3}\\
& & & \sum_{p \in \Pi(u,v)} \delta(p) f_p \leq \frac{n^{4/5}\tau}{2k} \cdot y_{uv} & \forall (u,v) \in D_b, \label{lp:poly_4}\\
& & & 0 \leq y_{u,v},f_p,x_e \leq 1 & \forall (u,v) \in D_b, p \in \Pi(u,v), e \in E. 
\end{align}
\end{subequations}

For every $(u,v) \in D_b$, let $\Pi(u,v)$ be the set of paths $P$ from $u$ to $v$ in $G$ such that $\sigma(P) \leq \tau/ n^{4/5}$ and $\res(P) \leq \bdgt(u,v)$ (i.e., the distance constraints are satisfied). Also, let $\Pi = \cup_{(u,v) \in D_b} \Pi(u,v)$. The Linear program \eqref{lp:thin_pair_original} attempts to find a cheap and feasible (resource constraint-wise) solution where all the paths are taken from $\Pi$. But the problem here is that this solution could have paths that are expensive in the \ppc $\delta$. We resolve this issue by careful processing based on \cite{antonakopoulos2009}.

Let $(\hat{x},\hat{y},\hat{f})$ be a feasible solution to LP \eqref{lp:thin_pair_original} whose value is within a $(1+\zeta)$ factor of the optimal for any fixed $\zeta > 0$ (see Section \ref{sec:lpsoln} for the method to obtain this solution). 

Because the optimal set of paths for the pairs in $D_b^3$ (i.e., $P_3^* = \{P_{uv}^* \mid (u,v) \in D_b^3\}$) corresponds to a basic feasible solution of LP \eqref{lp:thin_pair_original}, the objective value of $(\hat{x},\hat{y},\hat{f})$ is atmost $(1+\zeta) \cdot \tau$. Now, let $(\check{x},\check{y},\check{f}) = (2\hat{x},\hat{y},2\hat{f})$. Set $\check{f_p} = 0$ for every path $p$ such that $\delta(p) > n^{4/5} \tau /k$. Then we reduce other $\check{f_p}$ values so that $\sum_{p \in \Pi(u,v)} \check{f_p} = y_{u,v} \: \forall (u,v) \in D_b$, and also prune the $x_e$ values such that $\check{x}_e =   \max_{(u,v) \in D_b} \{\sum_{p:e \in P \in \Pi_{(u,v)}} \check{f_p}\}$ for all $e \in E$ (i.e., we prune $x_e$ to make it as small as it can be while ensuring it can handle the necessary flow). This new $(\check{x},\check{y},\check{f})$ is another basic feasible solution to LP \eqref{lp:thin_pair_original}.\nnote{I think this needs to be proved in at least a little bit of detail} 

For now, we have $\sum_{p \in \Pi(u,v)} \check{f_p} = y_{u,v} \: \forall (u,v) \in D_b$; in addition, any path that still has $\check{f_p} > 0$ has both cheap investment and cheap maintenance (i.e., $\sigma(P) \leq \tau/n^{4/5}$ and $\delta(P) \leq n^{4/5}\tau/k$). Furthermore, these paths also satisfy the distance constraint of the demand pair they connect. All we have to do now is round this solution.

\paragraph{Rounding our solution:}

Now we need to round the solution of LP \eqref{lp:thin_pair_original} appropriately to decide which paths we need to include in our final solution. The overall structure of our rounding procedure is similar to that of \cite{gkl2023}. We first round $\check{x}_e$ and use that to create a temporary graph $G_{\text{temp}}$. Then for each $(s,t) \in D$, we find the cheapest (in terms of \ppc) $s \leadsto t$ path $p(s,t)$ that satisfies $\res(p(s,t)) \leq  \bdgt(s,t)$ and has $\sum_{e \in p_{(s,t)}}{\delta(e) \leq L_2}$ (if such a path exists) and add it to the set $P_3$. Note that we can use Claim \ref{cl:horvath_modified_onerational} for this purpose. Then we show that this procedure resolves sufficiently many bad pairs with high probability. 

Note that once we round the edges, we do not worry about the \upc. Since we only need to pay once for \upc, the algorithm does not try to optimize for \upc (that is handled by LP \eqref{lp:thin_pair_original}).

Let $P_3$ be the set of paths obtained by running Algorithm \ref{alg:lp_rounding} on $\{\check{x}_e\}$ and $G_{temp}$ be the graph returned by the same. 

\begin{algorithm}[!htb]
\caption{Thin pair rounding [LP rounding] ($x_e,D$)} \label{alg:lp_rounding}
\begin{algorithmic}[1]
\State{$E'' \gets \phi$ .}

\For{$e \in E$}
\State $\text{ Add } e \text{ to } G_{\text{temp}} \text{ with probability} \min\{n^{4/5}\ln  n  \cdot x_e,1\};$
\EndFor

\For{$(s,t) \in D$}
\State  Find the cheapest path (in terms of \ppc) $p_{(s,t)}$ in  $G_{\text{temp}}$ that fulfills the distance constraints $\bdgt (s,t)$. Add the path to $P_3$ if it has \ppc $\leq L_2$.

\EndFor

\State \Return $P_3,G_{temp}$
\end{algorithmic}
\end{algorithm}

The below lemma is an adaptation of Claim 2.3 from \cite{berman2013approximation,gkl2023}. We only change the constants involved.  

\begin{claim} \label{le:thin_helper_1}
    Let $A \subseteq E$. If Algorithm \ref{alg:lp_rounding} receives a fractional vector $\{\check{x}_e\}$ with non-negative entries satisfying $\sum_{e\in A} \check{x}_e \geq 1/10$, then the probability that Algorithm \ref{alg:lp_rounding} returns a graph $G_{temp}$ that is disjoint from $A$ is $\leq \exp((-1/10) \cdot n^{4/5} \cdot \ln n )$.
\end{claim}

\begin{proof}
    Let $G_{temp} = (V_{temp},E_{temp})$.
    If $A$ does have an edge $e$ which has $\check{x}_e \geq 1/(n^{4/5}\ln n)$, then $e$ is clearly included in $E_{temp}$. 

    If that is not the case, then the probability that none of the edges in $A$ are included in $E_{temp}$ is
    \begin{equation}
        \prod_{e \in A}(1 - n^{4/5} \ln n \cdot \check{x}_e) \leq \exp \left( -\sum_{e \in A}n^{4/5} \ln n \cdot \check{x}_e\right) \leq \exp \left( - \frac{1}{10} n^{4/5} \ln n \right).
        \notag
    \end{equation}
\end{proof}

Just as in \cite{gkl2023,berman2013approximation}, we now define our variant of anti-spanners, which we will call as anti-buy at bulk spanners. Anti-buy at bulk spanners are a useful tool for the following part of our proof. Our definition here needs to be more general for it needs to account for length, and both \ppc as well as \upc. But the proof itself is very similar to \cite{gkl2023}.

\begin{definition} \label{def:anti}
A set $A \subseteq E$ is an anti-buy at bulk spanner for a terminal pair $(s,t) \in E$ if $(V, E \setminus A)$ contains no feasible path $s \leadsto t$ path of \upc at most $L_1$ and \ppc $L_2$. If there is no proper subset of an anti-buy at bulk spanner $A$ for $(s,t)$ which is also an anti-buy at bulk spanner for $(s,t)$, then we say that $A$ is minimal. We use $\mathcal{A}$ to denote the set of all minimal anti-buy at bulk spanners for all bad edges.
\end{definition}

We now bound the number of minimal anti-buy at bulk spanners across all bad pairs. The following claim is from \cite{gkl2023,berman2013approximation} - the proof is only given for the sake of completeness (and because we are dealing with a more general structure here). 

\begin{lemma} \label{le:thin_antispanner_bound}
    Let $\mathcal{A}$ be the set of all minimal anti-spanners for bad pairs. Then $|\mathcal{A}|$ is at most $|D|\cdot 2^{(n/\beta)^2/2}$.
\end{lemma}

\begin{proof}
    Let $PS(s,t)$ be the power set of all edges in the local graph for a specific bad pair $(s,t)$. Since $(s,t)$ is a bad pair we have at most $n/\beta$ vertices in the local graph. This also means that we have $\leq (n/\beta)^2/2$ edges in the local graph of $(s,t)$. Therefore if $(s,t)$ is a bad pair, then we have $|PS(s,t)| \leq {2^{(n/\beta)^2/2}}$. 
    
    Observe that every anti-buy at bulk spanner for a specific demand pair $(s,t) \in D$ is a set of edges. Therefore it corresponds to an element in $PS(s,t)$. Set $PS_{\text{ bad }} = \bigcup_{(s,t)} PS(s,t)$ where $(s,t) \in D$ are bad pairs. Thus, we can see that, $|\mathcal{A}| \leq |PS_{\text{ bad }}| \leq |D| \cdot {2^{(n/\beta)^2/2}}$ which proves our result.
\end{proof}

The following two lemmas will finally show that the density of the solution obtained after the rounding procedure is large enough.

\begin{lemma} \label{le:thin_helper_3}
    \label{le:berman_thin_settle_rounding}
    With high probability, the set of paths $P_3$ settles every bad pair $(s,t) \in D$ that has $\hat{y}_{s,t} \geq 1/10$. 
\end{lemma}

\begin{proof}
    For every bad pair $(s,t) \in D$ with $\check{y}_{s,t} \geq 1/10$, if $A$ is an anti-spanner for $(s,t)$ then $\sum_{e \in A} \check{x}_e \geq \sum_{P \in \Pi(s,t)} \check{f}_p = \check{y}_{s,t} \geq 1/10$. 

    By Claim \ref{le:thin_helper_1}, the probability that $A$ is disjoint from $G_{temp}$ is at most $\exp(-n^{4/5} \cdot \ln n /10)$. Then, using Lemma \ref{le:thin_antispanner_bound}, we can bound the number of minimal anti-spanners for bad pairs and then if we apply union bound, we have the probability that the graph $G_{temp}$ is disjoint from any anti spanner for a bad pair is at most
    \begin{equation} \label{eq:thin_helper_eqn1}
        \exp\left(-\frac{1}{10}n^{4/5} \cdot \ln n\right) \cdot |D|\cdot 2^{(n/\beta)^2/2}.
    \end{equation}

    Recall that  $|D| \leq n^2$. Since $\beta = n^{3/5}$, we have $(n/\beta)^2 = n^{4/5}$. Thus when we plug in the values, we get,

    \begin{equation} \label{eq:thin_helper_eqn2}
        \exp\left(-\frac{1}{10}\cdot n^{4/5} \cdot \ln n + \ln \left(n^2 \cdot 2 ^ {n ^ {4/5}/2}\right)\right) = \exp\left(-\Theta(n^{4/5} \ln n)\right).
        \notag
    \end{equation}

    This shows that the probability that our graph $G_{temp}$ is disjoint from any anti-spanner for any $(s,t) \in D$ where $(s,t)$ is a bad pair with $\hat{y}_{s,t} \geq 1/10$ is exponentially small. This means that $G_{temp}$ will have a feasible path with \upc $\leq L_1$, \ppc $\leq L_2$ for those $(s,t)$ pairs.This means that with high probability our set of paths resolves every bad pair $(s,t) \in D$ that has  $\hat{y}_{s,t} \geq 1/10$.
\end{proof}

\begin{lemma} \label{le:thin_cheap_final}
    For any $\eps > 0$, when $|D_b^2| \leq |D_b^3|$, with high probability, the density of $P_3$ is at most 
    \begin{equation}
        \tO(n^{4/5} \cdot \tau/k_b).
        \notag
    \end{equation}
\end{lemma}

\begin{proof}
    Notice that the expected cost of $P_3$ would be at most $n^{4/5 } \ln n \cdot \tau$. To see this note the expected cost due to \upc is at most $n^{4/5} \ln n \cdot \tau$. Furthermore, since we only add paths that have \ppc $\leq L_2 = n^{4/5} \tau /k$ and we only have $k$ units of demand, the total cost due to \ppc $\leq n^{4/5} \tau$.
      
    Now observe that the number of pairs $(s,t) \in D $ for which $\check{y}_{s,t} < 1/10$ is at most $k_b/6$. If that is not the case, then the amount of flow between all pairs is strictly less than $k_b/4$ and that violates constraint \eqref{lp:thin_poly_1}. From Lemma \ref{le:thin_helper_3}, all pairs for which have $\check{y}_{s,t} \geq 1/10$ will be resolved with high probability. This means that the expected density of $P_3$ is upper bounded by
    \begin{equation}
        \frac{n^{4/5 + \eps} \ln n \cdot \tau}{k_b/6} = \frac{6 n^{4/5+ \eps} \ln n \cdot \tau}{k_b} = \frac{\tO(n^{4/5 + \eps} \cdot \tau)}{k_b}.
        \notag
    \end{equation}
\end{proof}

\begin{proof}[Proof of Theorem \ref{thm:bbs45}] 
    With the help of Corollary \ref{cor:rcsprational}, we can settle all good pairs with high probability with cost  $\leq \tO(n^{4/5} \cdot \tau)$. For bad pairs, if we ever have $k_b \leq 4n^{6/5}$, we can use Corollary \ref{cor:bbs} to resolve them with cost $\leq \tO(n^{3/5 +\eps} \cdot \tau)$. Otherwise, we can make two sets of paths $P_2$ and $P_3$ using a junction tree and by rounding the modified solution to LP \eqref{lp:thin_pair_original}  respectively. By Lemmas \ref{le:thin_costly_final} and \ref{le:thin_cheap_final}, we can see that at least one of these sets of paths will have a density $\leq \tO(n^{4/5 + \eps} \cdot \tau /k_b)$. Just take the cheaper among them and keep repeating the process until we can resolve all bad pairs with a high probability. This process has total cost $\leq \tO(n^{4/5 + \ep} \cdot \tau)$.
\end{proof}

\subsection{LP solution} \label{sec:lpsoln}

This section is based on \cite{antonakopoulos2009,gkl2023}. We now describe how to solve LP \eqref{lp:thin_pair_original}. Note that LP \eqref{lp:thin_pair_original} has an exponential number of variables. So, we instead take the dual of this LP (shown in LP \eqref{lp:thin_pair_dual}) that has polynomially many variables and exponentially many constraints. If we have a valid separation oracle we can solve LP \eqref{lp:thin_pair_dual} using the ellipsoid method.

\begin{subequations} \label{lp:thin_pair_dual}
\begin{align}
& \max & & (k_b/4)\cdot \theta - \sum_{(u,v) \in D_b} \zeta_{(u,v)} \\
& \text{subject to}
& & \sum_{(u,v) \in D_b} \alpha_{e,(u,v)} \leq \sigma_e & \forall e \in E, \label{lp:poly_dual_1}\\
& & & \beta_{(u,v)} + \zeta_{(u,v)} \geq \theta + \frac{n^{4/5}\tau}{2k} \cdot \gamma_{(u,v)} & \forall (u,v) \in D_b,\label{lp:poly_dual_2}\\
& & & \beta_{(u,v)} \leq \sum_{e \in p} \alpha_{e,(u,v)} + \delta(p) \cdot \gamma_{(u,v)} & \forall(u,v) \in D_b,\forall p \in \Pi_{(u,v)}, \label{lp:poly_dual_3}\\
& & & \alpha_{e,(u,v)},\beta_{(u,v)},\gamma_{(u,v)},\zeta_{(u,v)},\theta \geq 0 & \forall e \in E,\forall (u,v) \in D_b,\forall p \in \Pi_{(u,v)}. 
\end{align}
\end{subequations}

We only have polynomially many constraints in \eqref{lp:poly_dual_1}, \eqref{lp:poly_dual_2}. Therefore it is straightforward to get a separation oracle for them. However, we have exponentially many constraints in \eqref{lp:poly_dual_3}. \cite{antonakopoulos2009} uses the restricted shortest path from \cite{hassin1992approximation} for this purpose. But because our set of paths $\Pi_{(u,v)}$ also has distance constraints to account for, and we also need to handle negative lengths, we need something better. We use Corollary \ref{cor:rcsprational} as our separation oracle. The following claim is very similar to Claim 2.13 in \cite{gkl2023}.

\begin{claim}
    For a specific $(s,t) \in D$, $2$-\rcsp is a separation oracle for those constraints in equation \eqref{lp:poly_dual_3}
\end{claim}

\begin{proof}
 The first constraint can check if $\beta_{(u,v)} > \sum_{e \in p} \alpha_{e,(u,v)} + \delta(p) \cdot \gamma_{(u,v)}$. 
 We can use the second resource constraint in  $2$-RCSP to ensure that the distance constraint for $(s,t)$ is satisfied. We can now try to find a minimum cost $s \leadsto t$ path in this instance of $2$-\rcsp. If the \upc obtained when we meet these constraints is less than $L_1$, then we have a violating constraint and if not we do not have one.

\end{proof}

Now, we see an approximate version of LP \eqref{lp:thin_pair_dual}.

\begin{subequations} \label{lp:thin_pair_approximate_dual}
\begin{align}
& \max & & (k_b/4)\cdot \theta - \sum_{(u,v) \in D_b} \zeta_{(u,v)} \\
& \text{subject to}
& & \sum_{(u,v) \in D_b} \alpha_{e,(u,v)} \leq \sigma_e & \forall e \in E, \label{lp:poly_dual_approximate_1}\\
& & & \beta_{(u,v)} + \zeta_{(u,v)} \geq \theta + \frac{n^{4/5}\tau}{2k} \cdot \gamma_{(u,v)} & \forall (u,v) \in D_b,\label{lp:poly_dual_approximate_2}\\
& & & \beta_{(u,v)} (1+\zeta) \leq \sum_{e \in p} \alpha_{e,(u,v)} + \delta(p) \cdot \gamma_{(u,v)} & \forall(u,v) \in D_b,\forall p \in \Pi_{(u,v)}, \label{lp:poly_dual_approximate_3}\\
& & & \alpha_{e,(u,v)},\beta_{(u,v)},\gamma_{(u,v)},\zeta_{(u,v)},\theta \geq 0 & \forall e \in E,\forall (u,v) \in D_b,\forall p \in \Pi_{(u,v)}. 
\end{align}
\end{subequations}

Since our $\delta(e)$ values are rational and our $\ell_e$ values are all integers, we can just use Corollary \ref{cl:horvath_modified_onerational} as a separation oracle for the exponentially many constraints in LP \eqref{lp:thin_pair_approximate_dual} and exactly solve it. The value of any solution we obtain this way would be $\leq (1+\zeta) \cdot \opt$ where $\opt$ is the optimal value of LP \eqref{lp:thin_pair_original} (see \cite{gkl2023} for details).

\section{Approximation for Buy-at-Bulk Spanners in Terms of $k$} \label{sec:spannerk}

This section is dedicated to proving the following theorem. Recall that in \eqref{def:eta}, \eqref{def:xi}, and \eqref{def:sign}:
\begin{equation*}
    \eta:= \frac{|\min\{\min_{e \in E}\{\ell_e\},0\}|}{\min_{(s,t) \in D} \{|\bdgt(s,t)|\}},
    \xi:= \frac{\max_{(s,t) \in D} \{|\bdgt(s,t)|\}}{\min_{(s,t) \in D} \{|\bdgt(s,t)|\}}, \text{ and }
\siign(x) = \begin{cases} 
                 -1  & \text{if } x < 0, \\
                 0 & \text{if } x = 0, \\
                 1 & \text{if } x > 0.
            \end{cases}
\end{equation*}

\thmbbsr*

\begin{proof}

We first introduce the notion of $\theta$-relaxed distance-constrained junction tree solution.

\begin{definition}
A \emph{$\theta$-relaxed distance-constrained junction tree solution} is a collection of distance-constrained junction trees rooted at different vertices, that satisfies \eqref{eq:theta-relaxed} for all $(s,t) \in D$. These junction trees are called \emph{$\theta$-relaxed distance-constrained junction trees}.
\end{definition}

We construct a $\theta$-relaxed distance-constrained junction tree solution and compare its objective with the optimal $\theta$-relaxed distance-constrained junction tree solution with objective value $\opt_{junc}$.

We show the existence of an $\beta$-approximate solution consisting of $\theta$-relaxed distance-constrained junction trees. Here, $\beta=O(\sqrt{k})$ for \bbs and $\beta=1$ for \ssbbs.

Let $\opt$ denote the cost of the optimal solution where the distance constraints are strict, $\opt_{\theta}$ denote the cost of the optimal solution where the distance constraints are relaxed as \eqref{eq:theta-relaxed}, and $\beta$ denote the ratio between $\opt_{junc}$ and $\opt_\theta$. Clearly, $\opt_\theta \le \opt$ because the distance constraints for $\opt$ is stricter. It suffices to show (constructively) that $\beta=O(\sqrt{k})$ for \bbs and $\beta=1$ for \ssbbs because Theorem \ref{thm:mdrcjttheta} implies the existence of an $\tO(\beta k^\ep)$-approximation algorithm in $\poly(n,1/\theta,\eta,\xi)$-time.

To show that $\beta=1$ for \ssbbs, let $H$ be an optimal solution. We observe that $H$ itself is a $\theta$-relaxed distance-constrained junction tree rooted at the source $s$ that is connected to all the $k$ sinks, so $\beta=1$.

To show that $\beta=O(\sqrt{k})$ for \bbs, we use a density argument via a greedy procedure which implies an $O(\sqrt{k})$-approximate $\theta$-relaxed distance-constrained junction tree solution. We recall that the density of a $\theta$-relaxed distance-constrained junction tree is its cost divided by the number of terminal pairs that it connects while satisfying \eqref{eq:theta-relaxed}.

Intuitively, we are interested in finding low-density $\theta$-relaxed distance-constrained junction trees. We show that there always exists a $\theta$-relaxed distance-constrained junction tree with density at most an $O(\sqrt{k})$ factor of the optimal density. The proof of Lemma~\ref{lem:sqrt-k-den} closely follows the one for the directed Steiner network problem in \cite{chekuri2011set}, pairwise spanners \cite{grigorescu2021online}, and weighted spanners \cite{gkl2023}, by considering whether there is a \emph{heavy} vertex that lies on $s \leadsto t$ paths for $(s,t) \in D$ or there is a simple path with low density. The case analysis also holds with $\theta$-relaxed distance constraints.

\begin{restatable}{lemma}{lemsqrtkden} \label{lem:sqrt-k-den}
There exists a $\theta$-relaxed distance-constrained junction tree $\cJ$ with density at most $\opt_\theta / \sqrt{k}$.
\end{restatable}

\begin{proof}
Let $\{p^*(s,t)\}_{(s,t) \in D}$ (a collection of $s \leadsto t$ paths) be the optimal \bbs solution with cost $\opt_\theta$ while considering \eqref{eq:theta-relaxed}. The proof proceeds by considering the following two cases: 1) there exists a vertex $r \in V$ that belongs to at least $\sqrt{k}$ $s \leadsto t$ paths that satisfy \eqref{eq:theta-relaxed} for distinct $(s,t)$, and 2) there is no such vertex $r \in V$.

For the first case, we consider the union of the $s \leadsto t$ paths, each satisfying its relaxed distance constraint \eqref{eq:theta-relaxed}, that passes through $r$. This forms a subgraph in $\{e \mid e \in p^(s,t), (s,t) \in D\}$ which contains an in-arborescence and an out-arborescence both rooted at $r$, whose union forms a $\theta$-relaxed distance-constrained junction tree. This distance-constrained junction tree has cost at most $\opt_\theta$ and connects at least $\sqrt{k}$ terminal pairs, so its density is at most $\opt_\theta / \sqrt{k}$.

For the second case, each vertex $r \in V$ appears in at most $\sqrt{k}$ $s \leadsto t$ paths in $\{p^*(s,t)\}_{(s,t) \in D}$. More specifically, each edge $e \in E$ also appears in at most $\sqrt{k}$ $s \leadsto t$ paths in $G$. By creating $\sqrt{k}$ copies of each edge with the same cost $\sigma$ and $\delta$, all terminal pairs can be connected by edge-disjoint paths. Since the overall duplicate cost is at most $\sqrt{k} \cdot \opt_\theta$, at least one of these paths has cost at most $\sqrt{k} \cdot \opt_\theta / k $. This path constitutes a distance-constrained junction tree whose density is at most $\opt_\theta / \sqrt{k}$.
\end{proof}

Consider an iterative procedure that finds a minimum density $\theta$-relaxed distance-constrained junction tree and continues on the remaining disconnected terminal pairs. Suppose there are $t$ iterations, and after iteration $j \in [t]$, there are $n_j$ disconnected terminal pairs. For notation convenience, let $n_0 = k$ and $n_t = 0$. After each iteration, the minimum cost for connecting the remaining terminal pairs in the remaining graph is at most $\opt_\theta \le \opt$, so the total cost of this procedure is upper-bounded by
\[
\sum_{j=1}^t \frac{(n_{j-1} - n_j)\opt}{\sqrt{n_{j-1}}} \leq \sum_{i=1}^k \frac{\opt}{\sqrt{i}} \leq \int_1^{k+1} \frac{\opt}{\sqrt{x}} dx = 2 \opt (\sqrt{k+1} - 1) = O(\sqrt{k}) \opt\]
where the first inequality uses the upper bound by considering the worst case when only one terminal pair is removed in each iteration of the procedure.
\end{proof}

\section{Minimum Density Distance-Constrained Junction Trees} \label{sec:jtree}

Recall Definition \ref{def:theta_junc_bbs_main}.

\defthetajuncbbs*

In this section, we prove Theorem \ref{thm:mdrcjttheta}.

\thmmdrcjttheta*

    \subsection{Proof of Theorem \ref{thm:mdrcjttheta}}
    
    \subsubsection{High-level Idea}
       
    We do the following procedure using every possible root $r \in V$ and then take whichever case has the minimum density among all the possible roots. 

    First, we scale all the edge lengths to restrict the number of values that any path length could have. We can do this because of the leeway allowed by $\theta$. Then, we build a layered graph to turn the scaled distance constraints into connectivity constraints. Then we eliminate the pay-per-use costs by turning them into one-time costs. To do this we use the variant of the height reduction lemma presented in \cite{cekp} and build another layered graph with much fewer layers which we transform again into a tree-like graph. We then exploit the tree-like structure of this newly reduced graph to eliminate the pay-per-use costs. This gives rise to an instance of the \mslc problem. Finally, we can solve the instance of \mslc (while keeping the distance constraints in mind)  by using the \gsf problem using the approach presented in \cite{chlamtavc2020approximating}.

    At a high level, we follow the overall flow of the proof of Theorem 5.1 in \cite{chlamtavc2020approximating} which is in turn based on the proof structure in \cite{chekuri2011set}. Several changes are required in this approach because we have to deal not only with negative numbers but also fractional lengths. \cite{chlamtavc2020approximating} only handles unit lengths. \cite{gkl2023} generalizes it to positive integers that are polynomial in $n$. Dealing with fractional lengths (that don't even need to be polynomial in $n$ when positive) requires us to go beyond the techniques in \cite{gkl2023}. And we also have to deal with \ppc which has never been considered with distance constraints. 

\subsubsection{Scaling the weights}

Recall that $\siign(x)$ is defined as follows:

\begin{equation*}
    \siign(x) = \begin{cases} 
                 -1  & \text{if } x < 0, \\
                 0 & \text{if } x = 0, \\
                 1 & \text{if } x > 0.
            \end{cases}
\end{equation*}

\begin{definition}
    We say that a $s \leadsto t$ path $p_{s,t}$ is {\em feasible} if the length of the path is at most $\bdgt(s,t)$.
    
\end{definition}

\begin{definition}
    
    We slightly abuse notation and say that $p_{s,t}$ is {\em $\theta$-feasible} if the length of $p_{s,t}$ is at most $\bdgt(s,t) \cdot (1+ \theta \cdot \siign(\bdgt(s,t)))$.

\end{definition}

Let the length of $P_{s,t}$ be denoted by $\ell(P_{s,t})$; we also use $\res(P_{s,t})$ for the same. Let $\ell_{max} =  \max_{(s,t) \in D}| {\bdgt(s,t)} |$ and $\ell_{min} = \min_{(s,t) \in D}  | {\bdgt(s,t)} |$. Also, let $\lenlow = \min_{e \in E} \ell_e \text{ and } \lenup = \max_{e \in E} \ell_e$. Let $\Delta = \theta \cdot (\ell_{min}) /(n-1)$: intuitively $\Delta$ is a measure of the level of precision we need for the edge lengths. Now, we scale all edge lengths in the following way: for any edge $e \in E$, we set $\Bar{\ell}_{e} = d_{e} \cdot \Delta$ where $d_{e}$ is some integer which ensures $(d_{e} - 1) \cdot \Delta < \ell_e \leq d_{e} \cdot \Delta$ is true. We call this new graph with the scaled edge lengths as the scaled graph $\bar{G}$. $\bar{G}$ has the vertex set $V$ and edge set $E$, but the edge lengths in $\bar{G}$ are set to $\bar{\ell}_e$. The \upc and \ppc for the edges in $\bar{G}$ are inherited from the corresponding edges in $G$.

In the following discussion, we are going to slightly abuse notation and use the same variable to denote a path in both the original and scaled graph. The idea is that the scaled version of a path from the original graph is another path with the corresponding sequence of vertices. The only thing that will change is the function used to calculate the length of the path in the original and scaled graph.

Let $\Bar{\ell}(P) = \sum_{e \in P} \Bar{\ell}_{e}$ (recall that ${l}(P) = \sum_{e \in P} \ell_e$) for any path $P$. For any $s \leadsto t \: | (s,t) \in D$ path $P_{(s,t)}$ we have, 

\begin{align} 
     \Bar{\ell}(P_{s,t}) = \sum_{e \in P_{s,t}} \Bar{\ell}_{e} \leq \sum_{e \in P_{s,t}} (\ell_e +\Delta) \leq \sum_{e \in P_{s,t}} (\ell_e) +(n-1) \cdot (\Delta) \\
\end{align}

Thus, we have,
\begin{align} \label{eq:horvathjunc_inequality}
     \Bar{\ell}(P_{s,t}) \leq \ell(P_{s,t}) + \theta \cdot \ell_{min} \leq \ell(P_{s,t}) + \theta \cdot  \bdgt(s,t) \cdot \siign(\bdgt(s,t))
\end{align}

Furthermore, since $\ell_e \leq \bar{\ell}_e$, we have,
\begin{align} \label{eq:horvathjunc_inequality2}
     \ell(P_{s,t}) \leq \Bar{\ell}(P_{s,t})
\end{align}

Roughly speaking, \eqref{eq:horvathjunc_inequality} and \eqref{eq:horvathjunc_inequality2} tell us that the lengths of the scaled paths are close enough to the lengths of their respective original paths. 

More formally, we have the following claim:

\begin{claim} \label{cl:junctree_scalingfeasible}
     If $P_{s,t}$ is a feasible solution for $(s,t) \in D$ (i.e., $\ell(P_{s,t}) \leq \bdgt(s,t)$), then $P_{s,t}$ is a $\theta$-feasible solution when we scale the weights (i.e., $\Bar{\ell}(P_{s,t}) \leq (1+\theta \cdot \siign(\bdgt(s,t))) \cdot \bdgt(s,t)$).

     Furthermore, if ${P_{s,t}}$ is a $\theta$-feasible solution in the scaled graph, then $\ell(P_{s,t}) \leq (1+\theta \cdot \siign (
     \bdgt(s,t))) \cdot \bdgt(s,t)$ in the unscaled/original graph.
\end{claim}
\begin{proof}

      If $P_{s,t}$ is a feasible solution for $(s,t) \in D$, then $\ell(P_{s,t}) \leq \bdgt(s,t)$. Thus, using \eqref{eq:horvathjunc_inequality}, we can see that $\bar{\ell}(P) \leq (1+\theta \cdot \siign(\bdgt(s,t))) \cdot  \bdgt(s,t)$. Thus, $P_{s,t}$ is a $\theta$-feasible solution for the scaled version too. 

      Now, let us look in the opposite direction. If $P_{s,t}$ is a $\theta$-feasible solution in the scaled graph, then $\bar{\ell}(P_{s,t}) \leq (1+\theta \cdot \siign(\bdgt(s,t))) \cdot  \bdgt(s,t)$. Then using \eqref{eq:horvathjunc_inequality2}, we can see that,

      \begin{align} 
        \ell(P_{s,t}) \leq \Bar{\ell}(P_{s,t}) \leq (1+\theta \cdot \siign(\bdgt(s,t))) \cdot  \bdgt(s,t)
    \end{align}
     which proves our Claim.
\end{proof}

Next, we prove another claim that compares the cost of a partial solution in $\bar{G}$ with the cost of a partial solution in $G$.

\begin{claim} \label{cl:junctree_scalingoptimal}
    For any $f > 0$, and set of terminal pairs $D' \subseteq D$, if there exists a set of paths $\{p_1(s,t)\}_{(s,t) \in D'}$ in $G$ of total cost (both \upc and \ppc) $\leq f$ containing a path of length at most $\bdgt(s,t)$ from $s$ to $t$ for every $(s,t) \in D'$ then there exists a set of paths $\{p_2(s,t)\}_{(s,t) \in D'}$ in $\bar{G}$ of total cost (both \upc and \ppc) $\leq f$ containing a path of length at most $\bdgt(s,t) (1+\theta \cdot \siign(\bdgt(s,t)))$ from $s$ to $t$ for every $(s,t) \in D'$.

    In addition, for any $f > 0$, if there exists a set of paths $\{p_1(s,t)\}_{(s,t) \in D'}$ in $\bar{G}$ of total cost (both \upc and \ppc) $\leq f$ containing a path of length at most $\bdgt(s,t) (1+\theta \cdot \siign(\bdgt(s,t)))$ from $s$ to $t$ for every $(s,t) \in D'$ then there exists a set of paths $\{p_2(s,t)\}_{(s,t) \in D'}$ in $G$ of total cost (both \upc and \ppc) $\leq f$ containing a path of length at most $\bdgt(s,t)(1+\theta \cdot \siign(\bdgt(s,t)))$ from $s$ to $t$ for every $(s,t) \in D'$.
\end{claim}
\begin{proof}
   The first part of the claim can be proved easily. We can just set $\{p_2(s,t)\}_{(s,t) \in D'} = \{p_1(s,t)\}_{(s,t) \in D'}$. Now, for any $(s,t) \in D$, $\{p_1(s,t)\}_{(s,t) \in D'}$ has some path $p_1(s,t)$ that is feasible. Then using the first part of Claim \ref{cl:junctree_scalingfeasible} we can see that $p_1(s,t)$ is a $\theta$-feasible path for the same $(s,t)$ in $\bar{G}$. Since we are using the same set of paths and costs are inherited, the cost will remain the same.

    For the second part, we can again set $\{p_2(s,t)\}_{(s,t) \in D'} = \{p_1(s,t)\}_{(s,t) \in D'}$. Now, for any $(s,t) \in D$, $\{p_1(s,t)\}_{(s,t) \in D'}$ has some path $p_1(s,t)$ that is $\theta$-feasible. Using the second part of Claim \ref{cl:junctree_scalingfeasible} we can see that $\ell(p_1(s,t))$ is a $\theta$-feasible path for the same $(s,t)$ in $G$. Since we are using the same set of paths and costs are inherited, the cost will remain the same.




\end{proof}

From now we will only be dealing with the scaled graph $\bar{G}$.  
When we build a layered graph as in \cite{chlamtavc2020approximating}, we will be building it on the scaled graph $\bar{G}$ instead of the original graph $G$.


\subsubsection{Turning distance constraints into connectivity constraints}

\paragraph{High level idea and potential challenges:}

For a specific root vertex $r$, we turn our distance constraints with \upc and \ppc problem into a connectivity problem with \upc and \ppc edges by building a layered graph. While the overall approach of turning a distance constraint into a connectivity constraint was introduced in \cite{chlamtavc2020approximating}, we have to keep several things in mind while designing a similar construction. 

\begin{itemize}
    \item The previous construction in \cite{chlamtavc2020approximating} only allows edges between successive layers - while this method is sufficient for unit length edges simply won't work/make sense in our model - 
    \begin{itemize}

        \item Our problem is not just dealing with unit lengths. One workaround is to decompose and edge $e$ with length $\ell_e$ into $\ell_e$ smaller edges all of unit lengths as in \cite{gkl2023} - this would fail in our model because our lengths are not integers and even worse they don't even need to be polynomial in $n$.
        \item This technique is not equipped to handle negative lengths, even $\{-1,1\}$. We need to somehow ensure that allowing negative lengths plays no role/obstacle in our overall problem. 
    \end{itemize}
    \item Any construction we make should not break the other parts of the proof or the requirements of the other model parameters. Our construction here should not make handling the pay-per-use costs impossible. The main technique to handle those would be using the version of height reduction lemma from \cite{cekp} - but we have to ensure we mostly resolve distance constraints before we apply that - because the height reduction lemma isn't equipped to handle any form of distance constraints. 
\end{itemize}

\paragraph{Graph construction:} Let us now see our graph construction which needs to keep all of these concerns in mind.  Let $t^{-}$ and $t^{+}$ represent the smallest and largest possible multiples of $\Delta$ that the length of any subpath of any feasible path could take.  $t^{-} = \lfloor \min (\lenlow \cdot (n-1) / \Delta , 0)  \rfloor$ - this happens when we take $n-1$ consecutive edges of scaled weight at least $\lenlow$ ; $t^{+} = \lceil \ell_{max} (1+\theta) / \Delta \rceil + | t^{-} | $. This is because any feasible scaled path has a length atmost $\ell_{max} \cdot (1+\theta)$  - but a subpath could be longer because it could decrease its length by as much as $t^{-} \cdot \Delta$ when it takes edges of negative length. 
We construct a layered graph $\bar{G}_r$ with the following vertices:
\begin{equation}
         \Bar{V}_r^L = \left( (V\setminus r) \times \{t^{-} \cdot \Delta, (t^{-}+1) \cdot \Delta, ,\ldots,(t^{+}-1)\cdot \Delta,t^{+} \cdot \Delta\} \times \{L\}\right) \cup \{(r,0)\},
\end{equation}

\begin{equation}
         \Bar{V}_r^R = \left( (V\setminus r) \times \{t^{-} \cdot \Delta, (t^{-}+1) \cdot \Delta, ,\ldots,(t^{+}-1)\cdot \Delta,t^{+} \cdot \Delta\} \times \{R\}\right) \cup \{(r,0)\},
\end{equation}

\begin{equation}
    \begin{split}
         \bar{V}_r = \bar{V}_r^R \cup \bar{V}_r^L
    \end{split}
\end{equation}

As an example, a vertex in the newly constructed graph looks as follows: $(u,I \cdot Delta, L)$. This denotes that the new vertex is a copy of the vertex $u$ from the scaled graph $\bar{G}$, and this vertex is in the $I^{th}$ layer from the root. Let us call this $I$ as the label of the layer. We will explain the relevance of $L$ and $R$ later on. For now, it suffices to think of them as two separate copies of the same vertex set.

We connect these vertices with the following edges:

\begin{equation}
    \begin{split}
        \Bar{E}_r^R = \{((u,I \cdot \Delta,R)(v,J \cdot \Delta,R)) | (u,I\cdot \Delta,R),(v,J\cdot \Delta,R) \in \bar{V}_r^R,(u,v)\in E \\
        \text{ and } \Bar{\ell}_{(u,v)}=(J-I )\cdot \Delta \text{ where } I,J \in  Z \text{ and } \Bar{\ell}_{(u,v)} \text{ is the scaled length of the edge} (u,v)\}
    \end{split}
\end{equation}

\begin{equation}
    \begin{split}
        \Bar{E}_r^L = \{((u,I \cdot \Delta,L)(v,J \cdot \Delta,L)) | (u,I\cdot \Delta,L),(v,J\cdot \Delta,L) \in \bar{V}_r^L,(u,v)\in E \\
        \text{ and } \Bar{\ell}_{(u,v)}=(J-I )\cdot \Delta \text{ where } I,J \in  Z \text{ and } \Bar{\ell}_{(u,v)} \text{ is the scaled length of the edge} (u,v)\}
    \end{split}
\end{equation}

\begin{equation}
    \Bar{E}_r = \Bar{E}_r^R \cup \Bar{E}_r^L
\end{equation}

Let, 
\begin{equation}
    \bar{G}_r = (\bar{V}_r,\bar{E}_r)v
\end{equation}

Intuitively, we add an edge between two vertices whenever it makes sense i.e. when the original copies of these two vertices are connected in the scaled graph $\bar{G}$ and the layer separation between these two vertices is equal to the length of the scaled edge in  $\bar{G}$. The edges in our layered graph $\Bar{G}_r$ inherit the upfront costs($\sigma(e)$) and pay-per-use costs($\delta(e)$) from the corresponding edges in the original graph. 

We call an integer $I$ {\em valid} if $t^{-} \leq I \leq t^{+}$. For every terminal pair $(s,t) \in D$, do the following, 

\begin{enumerate}
    \item Add new vertices $(s^t,I \cdot \Delta)$ and $(t^s,J \cdot \Delta)$ for all {\em valid} $I,J$ integers to $V_r$.
    \item For all such $I$ and $J$ add edges $((s^t,I \cdot \Delta)(s,(I\cdot \Delta,L)))$ and $((t,(J\cdot \Delta,R))(t^s,J\cdot \Delta))$ with zero \upc and \ppc to $E_r$.
    \item Now for every terminal pair $(s,t) \in P$ define:
    \begin{enumerate}
        \item terminal sets $S_{s,t} = \{(s^t,I \cdot \Delta) \:\forall \text{ {\em valid} integers } I\}$,
        \item $T_{s,t} = \{(t^s,J \cdot \Delta) | \:\forall \text{ {\em valid} integers }J\}$ and 
        \item relation $R_{s,t} = \{(s^t,I\cdot \Delta),(t^s,J\cdot \Delta) \in S_{s,t}\times T_{s,t} | (I+J) \cdot \Delta \leq \bdgt(s,t) \} $.
    \end{enumerate}
\end{enumerate}

Note that unlike \cite{chlamtavc2020approximating}, we have to scale the lengths first before building this layered graph. This scaling process necessitates several changes in our graph construction and also the proof.

\paragraph{Relating the layered graph with the scaled graph:}
Because of the construction of $\Bar{G}_r$, for every terminal pair $(s,t) \in D$ and valid integer $I$, there is a bijection between paths of length $I \cdot \Delta$ from $s$ to $r$ in $\bar{G}$, and paths from $(s^t,I\cdot \Delta)$ to $(r,0)$ in $\bar{G}_r,$ w. Similarly for every valid $J$, there is a bijection between paths of length $J\cdot \Delta$ from $r$ to $t$ in $\bar{G}$, and paths from $(r,0)$ to $(t^s,J\cdot \Delta)$ in $\bar{G}_r$. Now, to keep track of lengths in $\bar{G}$ we can just connect the appropriate terminal pairs in $\bar{G}_r$. 

\paragraph{Runtime:} 
Our runtime so far is polynomial in $|\bar{G}_r|$ (when we say $|\bar{G}_r|$, we mean the number of vertices in $|\bar{G}_r|$). Now, $|\bar{G}_r|$ is a polynomial in $O(n \cdot (| t^{-}  |+ | t^{+} |))$ and thus it is a polynomial in 

\begin{align}
    O(n \cdot (| \lenlow \cdot (n) | / \Delta + \ell_{max} \cdot (1+\theta) / \Delta) )  \\
    =  O(n \cdot (| \lenlow \cdot (n)^2 | / (\theta \cdot \ell_{min}) + n \cdot \ell_{max} \cdot (1+\theta) / (\theta \cdot \ell_{min}) ) )
\end{align}

Thus the overall runtime so far is a polynomial in $O(n^3 | \lenlow | / \theta \cdot \ell_{min} + n \cdot \ell_{max} \cdot (1+\theta) )/(\theta \cdot \ell_{min}))$.  Thus, when $1/\theta$, $| \lenlow |/ \ell_{min}$ and $\ell_{max}/\ell_{min}$ are polynomials in $n$, the runtime so far is a polynomial in $n$.

\paragraph{To summarize,} we have turned all budget constraints into connectivity constraints so far. We keep track of the budget constraints using some relations. We still have to deal with the \ppc. 

\subsubsection{Handling pay-per-use costs $\delta(e)$}

\paragraph{High level overview and tools used} We now need to account for the effects of the pay-per-use cost $\delta(e)$ for each edge. While we are doing this, we should not disrupt our previous work for handling the distance constraints - this means we do not disrupt the terminals in any way and retain ways to keep track of which of those need to be connected. 

We now present the height reduction lemma from \cite{cekp} here which is in turn a modified version of the work of \cite{helvig2001improved}. To improve clarity we use the terms \upc and \ppc to denote the cost of adding an edge and the cost of using it once in a path respectively.

\begin{restatable}{lemma}{lemmahrbb}
    \label{le:height_reduction_bb} \cite{cekp}
    (\textbf{Height Reduction}) Given a directed graph $G=(V,E)$ with upfront costs $\sigma(e)$ and pay-per-use costs $\delta(e)$, for all $h>0$, we can efficiently find an upward directed, layered graph $G_r^{up}$ on $(h+1)$ levels and edges (with new upfront and pay-per-use costs) only between consecutive levels going from bottom (level $h$) to top (level $0$), such that each layer has $n$ vertices corresponding to the vertices of $G$, and, for any set of terminals $X$ and any root vertex $r$,
    \begin{itemize}
        \item  the optimal objective value of the single-sink buy-at-bulk problem to connect $X$ (at level $h$) with $r$ (at level $0$) on the graph $G^{up}_r$ is at most $O(hk^{1/h})\rho$, where $\rho$ is the objective value of an optimal solution of the same instance on the original graph G,
        
        \item given an integral (fractional solution) of objective value $\rho$ for the single-sink buy-at-bulk problem to connect $X$ with $r$ on the graph $G^{up}_r$, we can efficiently recover an integral (fractional solution) of objective value at most $\rho$ for the problem on the original graph $G$.
    \end{itemize}
    In the same way, we can obtain a downward directed, layered graph $G^{down}_r$ on $(h+1)$-levels with edges going from top to bottom, satisfying the same properties as above except for single-source (as opposed to single-sink) instances instead.
\end{restatable}

Thankfully Lemma \ref{le:height_reduction_bb} is well equipped for our situation - so we use it right now (unlike \cite{chlamtavc2020approximating} which uses it only after formulating the \mslc instance).

\paragraph{Applying Lemma \ref{le:height_reduction_bb}:}  We first apply Lemma \ref{le:height_reduction_bb} on $(\bar{G}_r^R = (\bar{V}_r^R,\bar{E}_r^R))$ and $(\bar{G}_r^L = (\bar{V}_r^L,\bar{E}_r^L))$ seperately, and obtain two new $(h+1)$ layered graphs $G^{up}_r$ and $G^{down}_r$ where $h$ is some positive integer which depends on $1/\eps$. Unlike \cite{cekp} our input graph is going to be significantly more complicated than the base graph $G$ but it will have only one root $r$. We do not keep track of the intermediate vertices from the original graphs, and we don't need to do so either. But we can keep track of the terminal vertices and that is important.

\paragraph{Reduction to a layered tree:} Once we obtain this layered graph, we create two new tree-like graphs $\bar{T}^{up}_r$ and $\bar{T}^{down}_r$. The purpose of this tree-like graph is to ensure that we have only one path from a terminal to the root - which allows us to easily handle \ppc. We describe the construction of $\bar{T}^{up}_r$. The construction of $\bar{T}^{down}_r$ is similar (things need to be inverted appropriately):

\begin{itemize}
    \item The $0^{th}$ layer of $\bar{T}_r^{up}$ has just one vertex and it is the root $r$.
    
    \item For each $i$ such that $1 \leq i \leq h$, the $i^{th}$ layer of $\bar{T}_r^{up}$ contains all $(i+1)-$ length tuples $(r,v_1,\ldots,v_i)$ where $v_j$ is a vertex present in the $j^{th}$ layer of $\bar{G}^{up}_r$.
    
    \item For every edge $e = (v_i,v_{i-1}) \in \bar{G}^{up}_r$, there is an arc from $(r,v_1,\ldots,v_{i-1},v_i)$ to $(r,v_1,\ldots,v_{i-1})$ inheriting the same uprfront cost $\sigma(e)$ and pay-per-use cost $\delta(e)$ as in $\bar{G}_r^{up}$.
\end{itemize}

\paragraph{Eliminating \ppc:} After we get the layered tree, we create new terminal vertices $(s^t_r,I \cdot \Delta) \: (\forall \: (s,t) \in D, \text{ valid integers } I)$ and connect them to any leaf in $\bar{T}_r^{up}$ that is of the form $(r,v_1,v_2,\ldots,v_h)$ with $v_h=(s^t,I\cdot \Delta)$. We do the same by adding $(t^s_r,J\cdot \Delta)$ in $\bar{T}_r^{down}$ in the same way (but these vertices are sinks not sources). We then add an edge from $r$ in $\bar{T}_r^{up}$ to $r$ in $\bar{T}_r^{down}$ and call this final graph as $\bar{T}_r$. Because the $G^{up}_r$ has only $h$ layers, the newly created graph $\bar{T}^{up}_r$ has size only $O(|G^{up}_r|^{O(h)})$ which is polynomial in input size as $h$ is a constant and $|G^{up}_r|$ is also polynomial in input size. 

We now exploit the tree-like structure of our new graph where there is only one path from a leaf in $\bar{T}_r^{up}: (r,v_1,v_2,\ldots,v_h)$ to the root $r$. Let $v_h = (s^t,I \cdot \Delta)$ be a leaf node. Also, let $\eta(s^t,(r,v_1,v_2,\ldots,v_h)) = \sum_{e \in P} \delta(e)$ where $P$ is the only path that connects $(r,v_1,v_2,\ldots,v_h)$ to $r$. We set the \upc of the $((s^t_r,I\cdot \Delta)(r,v_1,v_2,\ldots,v_h))$ edges to $\eta(s^t,(r,v_1,v_2,\ldots,v_h)) \cdot \dem(s,t)$ and do something similar for the $(t^s_r,J\cdot \Delta)$ terminal vertices. Effectively, we are using the \upc of these newly created edges to capture the \ppc one needs to pay to use a specific path.  

After this step, we no longer need to do anything for the pay-per-use cost. $\bar{T}_r$ can discard the \ppc value stored in all edges (they are taken care of by the \upc in the terminal edges).

\paragraph{To summarize:} We use Lemma \ref{le:height_reduction_bb} to first restrict the number of layers one needs to traverse from a terminal to the root. Then we apply another reduction to ensure that there is only one path from the terminals to the root. To this, we add some dummy vertices that can capture the \ppc ($\delta(e)$) so that we no longer need to keep track of those in our main graph.

\subsubsection{Reduction to \mslc}

We create another much simpler graph: this one is composed of two graphs $G^{1}$ and $G^{2}$ which are both copies of $\bar{G}$, that intersect only in the vertex $r$. Let us call this graph $\hat{G}_r$. In addition, for every node $u \in V$, we use $u^{1}$ and $u^{2}$ to denote the copies of $u$ in $G^{1}$ and $G^{2}$. The following lemma (which is also in \cite{chlamtavc2020approximating}) follows from our construction. It relates our scaled problem to a problem in the graph $\bar{T}_r$. In simple words, it says for a partial solution of cost $f$ in the $\hat{G}_r$, we will have a partial solution (satisfying the same set of pairs) of cost $O(k^\eps \cdot f)$ in $\bar{T}_r$. In addition, it also says that for a partial solution of cost $f$ in the $\bar{T}_r$, we will have a partial solution of cost $\leq f$ in $\hat{G}_r$.

\begin{lemma} \label{le:junction_tree_helper}
    For any $f > 0$, and set of terminal pairs $D' \subseteq D,$ there exists a set of paths $\{p_1(s,t)\}_{(s,t) \in D'}$ in $\hat{G}$ of total cost (both \upc and \ppc) $\leq f$ containing a path of length at most $\bdgt(s,t)$ from $s^{1}$ to $t^{2}$ for every $(s,t) \in D'$ only if there exists a subgraph $\hat{Jun} \subseteq E(\bar{T}_r)$ of total weight $\leq f$ such that for every terminal pair $(s,t) \in D',$ $\hat{Jun}$ contains terminals $(s^t_r,I \cdot \Delta),(t^s_r,J \cdot \Delta)$ such that $((s^t,I\cdot \Delta)(t^s,J\cdot \Delta)) \in R_{s,t}$. Moreover, given such an edge set $\hat{Jun}$, we can efficiently find a corresponding edge set of paths $\{p_1(s,t)\}_{(s,t) \in D'}$ in $\hat{G}$. 

    For any $f > 0$, and set of terminal pairs $D' \subseteq D,$ if there exists a set of paths $\{p_1(s,t)\}_{(s,t) \in D'}$ in $\hat{G}$ of total cost (both \upc and \ppc) $\leq f$ containing a path of length at most $\bdgt(s,t)$ from $s^{1}$ to $t^{2}$ for every $(s,t) \in D'$ then there exists a subgraph $\hat{Jun} \subseteq E(\bar{T}_r)$ of total weight $\leq O(k^\eps \cdot f)$ such that for every terminal pair $(s,t) \in D',$ $\hat{Jun}$ contains terminals $(s^t_r,I \cdot \Delta),(t^s_r,J \cdot \Delta)$ such that $((s^t,I\cdot \Delta)(t^s,J\cdot \Delta)) \in R_{s,t}$. 
\end{lemma}

 Let $w_r$ be the \upc on the graph $\bar{T}_r$. Also, define $S_{s,t}^r = \{(s_r^t,I \cdot \Delta) | ((s_r^t,I\cdot \Delta),(r,v_1,\ldots,v_h))\in E(\bar{T}_r) : v_h = (s^t,I\cdot \Delta) \in S_{s,t}\} $ (in other words $S_{s,t}^r$ contains all terminals in $\bar{T}_r$ that correspond to the terminals in $S_{s,t}$) and  $T_{s,t}^r = \{(t_r^s,I\cdot \Delta) \mid ((t_r^s,I\cdot \Delta),(r,v_1,\ldots,v_h))\in E(\bar{T}_r) : v_h = (t^s,I\cdot \Delta) \in T_{s,t}\} $. Finally, set $R_{s,t}^r = \{(s^t_r,I\cdot \Delta),(t^s_r,J\cdot \Delta) \in S_{s,t}^r\times T_{s,t}^r \mid I\cdot \Delta+J\cdot \Delta \leq \bdgt(s,t) \} $. 

We now define the \mslc problem.

\begin{definition}
    In the \mslc problem, we have a directed graph $G = (V,E),$ two collections of disjoint vertex sets $\hat{S},\hat{T} \subseteq 2^V$, a collection of set pairs $P\subseteq \hat{S} \times \hat{T},$ and for each set pair $(S,T) \in P,$ a relation $R(S,T) \subseteq S \times T$ and non-negative edge costs $c: E \to \R_{\geq 0},$. The objective here is to find an edge set $F\subseteq E$ that minimizes the ratio
    \begin{equation*}
        \frac{\sum_{e \in F}c(e)}{|\{(S,T)\in P \mid \exists (s,t) \in 
 R(S,T):\text{ $F$ contains an $s \leadsto t$ path }\}|}
    \end{equation*}
\end{definition}

Now, to prove Lemma \ref{le:junction_tree_helper}, we just need to show that we can achieve an $O(n^\ep)$ approximation for the \mslc instance $(\bar{T}_r,\{S_{s,t}^r, T_{s,t}^r, R_{s,t}^r \mid (s,t) \in D\},w_r)$ obtained from our reduction.  

Thus, it suffices to show the following lemma.

\begin{lemma} \label{le:junction_tree_intermediate}
    In the given setting, there exists a $O(\log^3 n)$ approximation algorithm for the following problem that runs in polynomial time.

    \begin{itemize}
        \item Find a tree $T \subseteq T_r$ minimizing the ratio
        \begin{equation*}
            \frac{\sum_{e \in F}w_r(e)}{\mid\{(s,t)\in P \mid \exists (\hat{s},\hat{t}) \in 
 R_{s,t}^r:\text{ T contains an $\hat{s}$-$\hat{t}$ path }\}|}.
        \end{equation*}
    \end{itemize}
\end{lemma}
\begin{proof}
    This lemma has been proved in \cite{chlamtavc2020approximating} (we make no significant changes to the relation used in \cite{chlamtavc2020approximating}; we would need a different proof if we had a different relation).
\end{proof}

We will now see the proof of Corollary \ref{thm:mdrcjt}.

\thmmdrcjt*

\begin{proof}
    Note that when all $\ell_e \in [\poly(n)]_{\pm}$, we can assume without loss of generality that $\bdgt(s,t) \in [\poly(n)]_{\pm} \:\: \forall (s,t) \in D$. This means that $\eta,\xi \in [\poly(n)]_{\pm}$. Thus, using Theorem \ref{thm:mdrcjttheta}, when we set $\theta = 1/n \cdot \poly(n) \leq 1/\poly(n)$, our lemma is proved. 
\end{proof}

\section{Resource-constrained Shortest Path} \label{sec:rcsp}

In this section, we modify the result presented in \cite{horvath2018multi} to allow the lengths to be negative ( upto a specific range). Recall Definition \ref{def:rcsp}.

\defnrcsp*

Note that the only difference between this problem and the one solved in \cite{horvath2018multi} is that we allow the resource consumptions to be negative. This however makes many things more complicated.

Given resource constraints $L_1,L_2,\ldots,L_m$ for the \rcsp, let $\opt_{RCSP}$ be the cost of any minimum cost $s \leadsto t$ path that satisfies the resource constraints. An $(1; 1 + \eps_1, 1+\eps_2 . . . , 1 + \eps_m)$-approximation scheme finds an $s \leadsto t$ path whose cost is at most $\opt_{RCSP}$, but the $i^{th}$ resource constraint is satisfied up to a factor of $1 + \eps_i$ for that path. We will now present a $(1;1+\eps_1, 1+\eps_2,\ldots\,1+\eps_m)-$ FPTAS for $m$-\rcsp (under certain assumptions), where the number of weight functions, m, is a constant. Let $\boldsymbol{\epsvec}$ be a vector that is composed of $\eps_i \: \forall \: i \in [1,2,\ldots,m]$. 

Also, let $\boldsymbol{L}$ be a vector composed of all $L_i \: \forall \: i \in [1,2,\ldots,m]$ .Our problem can be reformulated as follows:

\begin{align} \label{eq:horvath_original}
    \arg \min_{P} \{c(P):w_i(P) \leq L_i, i=1,\ldots,m\}.
\end{align}  

One very important point that we wish to make right now is that any $\eps_i$ can be negative. The role of $\eps_i$ is to allow some leeway/approximation for our algorithm so that we can run it in polynomial time. Consider the case where $R_i < 0$ for some $i$. In this case, if $R_i<0$ then $\eps_i \cdot R_i < 0$ - which means that an algorithm has to use fewer resources than the allocated budget and get a cost at least as good as the optimal value for the allocated budget. This is clearly not possible in general. To allow some approximation, we need $\eps_i \cdot R_i > 0$ and when $R_i < 0$, we need $\eps_i <0$.

Given two vectors $\boldsymbol{a},\boldsymbol{b} \in Q^m$ we use $(\boldsymbol{a} \cdot \boldsymbol{b})$ to denote the Hadamard product vector $\boldsymbol{c}$ with $c_i = (a_i  b_i), 1\leq i \leq m$. We also use $\boldsymbol{a^{-1}}$ to denote a vector $\boldsymbol{c}$ with  $c_i = (1/a_i), 1\leq i \leq m$. Given two vectors $\boldsymbol{a}$ and $\boldsymbol{b}$, we say that $\boldsymbol{a} \leq \boldsymbol{b} \: (\boldsymbol{a}<\boldsymbol{b})$ if $a_i \leq b_i \: (a_i < b_i) \forall i \in \{1,\ldots,m\}$. Given a vector $\boldsymbol{a}$ and a real number $b$, we use $\boldsymbol{a} \cdot b$ to denote another vector $\boldsymbol{c}$ with $c_i = (a_i b), \: 1\leq i \leq m$ . Also, let $\maxni$ denote the $| \min(\min_{(u,v) \in E} w_i(u,v),0) |$. Finally, let $\boldsymbol{\maxn}$ be a vector of all the $\maxni$ for $i \in \{1,2,\ldots,m\}$.

The standard approach here for a specific $\boldsymbol{\epsvec}$ is to scale and round the given lengths. It can be shown that solving the same problem on the scaled version will give us an approximate solution to the original problem. Then a dynamic program approach is given to solve the problem using the scaled and rounded lengths. 

Intuitively the standard idea is to take advantage of the additional $\boldsymbol{\epsvec} \cdot \boldsymbol{L}$ budget available to us. We split this $\boldsymbol{\epsvec} \cdot \boldsymbol{L}$ into $n$ pieces for each edge in a path effectively allowing us to use up slightly more budget for each edge. Thus, it suffices to approximately track the edge weights since we are allowed some leeway for each of them.

\subsection{Scaling Procedure}

Throughout this section, we assume that the graph contains no negative cycle for any dimension (i.e., you can never find a path $P$ starting from and ending at $s$ that has $w_i(P) < 0$ for any $i$.

For a given an $\boldsymbol{\epsvec}$ (where $\boldsymbol{\epsvec}$ is a vector), the scaling and rounding procedure is as follows: i) for all $i=1,\ldots,m$ the scale vector $\boldsymbol{\Delta} \in \mathbb{Q}^m$ is given by, $\Delta_i = \eps_i L_i/(n-1)$ ii) for each edge $e \in E$, let $\bar{\boldsymbol{w}}(e) \in \mathbb{R}^m$ be the scaled weight of the edge $e$, defined as $\bar{w}_{i}(e) = d_i(a) \cdot  \Delta_i$  where $d_i(a)$ is some integer such that $(d_i(a) - 1) \Delta_i < w_i(a) \leq d_i(a) \Delta_i$ is satisfied.

Now, our scaled problem can be formulated as:

\begin{align} \label{eq:horvath_scaled}
    \arg \min \{c(P):\bar{w}_i(P) \leq (1+\eps_i) L_i, i=1,\ldots,m\}. 
\end{align}  

For any $u-v$ path $P$, for any $i=1,\ldots,m$ we have  

\begin{align} \label{eq:horvath_inequality}
    \bar{w}_i(P) = \sum_{e \in P} \bar{w}_i(e) \leq \sum_{e \in P} (w_i(e) + \Delta_i) \leq \sum_{e \in P} (w_i(e)) + (n-1)\Delta_i = w_i(P) +\eps_i L_i,
\end{align}

Note that because we do not allow negative cycles, we can be assured that any path here will only have $n-1$ edges (there is no reason to include a cycle).

The following lemma helps us relate the scaled version (\eqref{eq:horvath_scaled}) to the original problem (\eqref{eq:horvath_original}).

\begin{lemma} \label{le:horvath_conversion}
    If \eqref{eq:horvath_original} has a feasible solution, then \eqref{eq:horvath_scaled} will have a feasible solution. Further, in that case, any optimal solution for \eqref{eq:horvath_scaled} is a $(1;1+\eps_1,1+\eps_2,\ldots,1+\eps_m)-$ approximate solution for \eqref{eq:horvath_original}.
\end{lemma}

\begin{proof}
    If $P$ is a feasible solution for \eqref{eq:horvath_original}, then $w_i(P) \leq L_i$ is satisfied for all $i = 1,\ldots,m$. Then using \eqref{eq:horvath_inequality}, we can see that $\bar{w}_i(P) \leq (1+\eps_i) \cdot  L_i \: \forall i \in \{1,\ldots,m\}$. Thus, $P$ is a feasible solution for the scaled version \eqref{eq:horvath_scaled} too. 
    
Note that any feasible and optimal solution for \eqref{eq:horvath_scaled} will be a $(1;1+\eps_1,\ldots,1+\eps_m)$ solution for \eqref{eq:horvath_original}. The feasibility is preserved because none of the weights will increase when we change from $\bar{w_i}(e)$ to $w_i(e)$ and in addition the budget for \eqref{eq:horvath_scaled} is $(1+\boldsymbol{\epsvec})$ times the budget for \eqref{eq:horvath_original}. The cost remains optimal because the edge costs are identical in the scaled and unscaled versions (only the edge weights are scaled, not costs) and the scaled version is a less restricted problem than the unscaled version (since paths are allowed a greater budget).
\end{proof}

\subsection{A dynamic program approach to solve \eqref{eq:horvath_scaled}}

We use the term {\em pattern} to denote a vector $\boldsymbol{\eta} = (\eta_1,\ldots,\eta_m)$, where $\eta_i \: \forall \: (i = 1,\ldots,m)$ is an integer. Note that for any path $P$, there is some pattern $\boldsymbol{\eta}$ which has $\bar{\boldsymbol{w}}(P) = \boldsymbol{\eta} \cdot \boldsymbol{\Delta}$. We define a pattern $\eta$ to be {\em feasible} if $\boldsymbol{\eta} \cdot \boldsymbol{\Delta} \leq (| 1+\boldsymbol{\epsvec} |) \cdot \boldsymbol{L} +|  \boldsymbol{\maxn} \cdot n | $ is satisfied\footnote{Note that we are overestimating our upper bounds. But this wouldn't have a significant impact on our performance}. In our dynamic program, we only need to deal with feasible patterns - this is because paths that aren't feasible will violate the budget constraints and such paths will never be the solution. Note that we are adding an extra term: $| \boldsymbol{\maxn} \cdot n |$ to the budget here when accounting for feasible paths. The reason is that a path that momentarily goes over budget could get back to being under budget by taking a series of negative weight edges.

We call a pattern $\boldsymbol{\eta}$ to be {\em valid} if it is feasible and has $\boldsymbol{\eta} \cdot \boldsymbol{\Delta} \geq \boldsymbol{\maxn} \cdot n $. Valid patterns are the only patterns our DP algorithm needs to consider. This is because there is no way any path can have a weight less than $ \boldsymbol{\maxn} \cdot n $ since a path can have at most $n$ edges each of scaled weight at least $\boldsymbol{\maxn} $ (scaled weight of an edge is at least the original weight of the edge). Furthermore, paths have to be feasible to fulfill budget requirements.

\begin{lemma} \label{le:horvath_pattern_bound}
    The number of valid patterns is $O\bigl(\Pi_{i=1}^{i=m} \left(n+n/\eps_i+ n \cdot \maxni/(\eps_i L_i)\right) \bigr)$.
\end{lemma}

\begin{proof}
    Note that for any valid pattern $\boldsymbol{\maxn} \cdot n \cdot \leq \boldsymbol{\eta} \cdot \boldsymbol{\Delta} \leq (| 1+\boldsymbol{\epsvec} |) \cdot \boldsymbol{L} + | \boldsymbol{\maxn} \cdot n | $. Thus, for a specific dimension $i$, we would need only $(| 1+\eps_i | )L_i/\Delta_i + 2 n \cdot |\maxni/| \Delta_i = (n-1)(1+1/| \eps_i |)  + 2 n^2 \cdot |\maxni|/(| \eps_i | L_i) = O(n + n/\eps_i) + O(n^2 \cdot |\maxni|)/(| \eps_i | \cdot L_i) $ patterns overall.

    Thus, in total the number of valid patterns would be $O(\Pi_{i=1}^{i=m})(O(n + n/| \eps_i |) + O(n^2\cdot |\maxni|)/(| \eps_i | L_i) )$.
\end{proof}

Lemma \ref{le:horvath_pattern_bound} gives us the following corollary.

\begin{corollary} \label{cor:horvath_pattern_bound}
    If $\forall i \in \{1,2,\ldots,m\} \:\: |\maxni| / | L_i |$ and $1/| \eps_i |$ are polynomials in $n$, then the number of valid patterns is a polynomial in $n$. 
\end{corollary}

In \cite{horvath2018multi}, the approach is to use dynamic programming to find out the lowest cost one can pay to reach a specific vertex $v \in V$ for a specific feasible pattern. This works because any path's weight always increases when we add an edge.

We cannot directly use the algorithm presented in \cite{horvath2018multi} here because we now have negative weights. This means that the dynamic program might need to use values that haven't been solved yet. We get around this new issue by adding another dimension to track hop count (i.e., the number of edges we have seen in the path so far). This hop count will be non-decreasing when we add a new edge \footnote{On some level this is similar to what the Bellman Ford algorithm does.}.

In Algorithm \ref{alg:horvath_modified}, $DP(v,\boldsymbol{\eta},h)$ represents the least cost for any path to reach the vertex $v$ from the source $s$ within a budget of $\alpha = \boldsymbol{\eta} \cdot \boldsymbol{\Delta}$ and in less than or equal to $h$ hops. Let  $H = \{\boldsymbol{\eta}^1,\boldsymbol{\eta}^2 \ldots, \boldsymbol{\eta}^{|H|}\}$ be the set of valid patterns. The patterns in $H$  are partially ordered by the element-wise comparison we saw earlier. That is, if $\boldsymbol{\eta}^p \leq \boldsymbol{\eta}^q$ (this is done by comparing each element of $\boldsymbol{\eta}^p$ with the corresponding element of $\boldsymbol{\eta}^q$) is satisfied for two patterns $\boldsymbol{\eta}^p$ and $\boldsymbol{\eta}^q$, then $p \leq q$.

Note that for the scaled problem in \eqref{eq:horvath_scaled}, we only need to look at budgets which can be expressed as $\boldsymbol{\eta} \cdot \boldsymbol{\Delta}$ for some $\boldsymbol{\eta} \in H$ - this is because the resource consumptions of the edges are integral multiples of $\boldsymbol{\Delta}$. And any path in the scaled graph is a combination of the scaled edges, so the weight of any path can also be expressed as $\boldsymbol{\eta} \cdot \boldsymbol{\Delta}$.

We also have another DP matrix $PATH $ to track the path we need to use. $PATH(v,\boldsymbol{\eta},h)$ gives the previous vertex one should reach at to reach the vertex $v$ from $s$ within a budget of $\boldsymbol{\eta} \cdot \boldsymbol{\Delta}$ and within $h$ hops. We can backtrack from $PATH(t,\lfloor (1+\boldsymbol{\epsvec)} \cdot \boldsymbol{L} \cdot \boldsymbol{\Delta}^{-1} \rfloor,n)$ to retrieve our solution.

Here is a brief overview of how Algorithm \ref{alg:horvath_modified} works. Let any pattern that has at least one negative element be called a {\em negative pattern}. Patterns without any negative element are called {\em non-negative patterns}. We first set the cost of reaching the source vertex as 0 for any {\em non-negative pattern}  $\boldsymbol{\eta}$ and hop count $h$. Since we don't have any negative cycles we don't have to worry about reaching the source with a negative weight path. We then set the cost of reaching any other vertex $v$ in the graph to be infinity for all valid hop counts and budgets. After this, we find the minimum cost path for all budgets for each potential destination in the graph in increasing order of allowed hop count. For a specific hop count and a specific destination, we look at all incoming arcs and pick the cheapest path one can construct to this destination. Note that this path has fulfill both budget and hop count restrictions.

\begin{algorithm}[!htb]
\caption{Dynamic programming algorithm for \eqref{eq:horvath_scaled}} \label{alg:horvath_modified}
\begin{algorithmic}[1]

\State{$DP(s,\boldsymbol{\eta},h) \gets 0 \: (\forall \text{ non negative and valid } \boldsymbol{\eta}  \text{ and for all hop count) } h \in \{0,1,\ldots,n\}$.}

\State{$DP(s,\boldsymbol{\eta},h) \gets \infty \: (\forall \text{ negative and valid } \boldsymbol{\eta}  \text{ and for all hop count) } h \in \{0,1,\ldots,n\}$.}
\State{$DP(v,\boldsymbol{\eta},h) \gets \infty (\forall v \neq s, \text{ valid } \boldsymbol{\eta} \text{ and for all hop count) } h \in \{0,1,\ldots,n\}$.}

\For{$i = 1,2,\ldots,n$} \Comment{$i$ gives the hop count}
\For{$\boldsymbol{\eta} \in \boldsymbol{\eta}^1, \boldsymbol{\eta}^2, \boldsymbol{\eta}^3, \ldots , \boldsymbol{\eta}^{|H|}$}
\For{$v \in V$}
\For{$e \in \{(u,v) \in E: \bar{\boldsymbol{w}}(u,v) \leq \boldsymbol{\eta} \cdot \boldsymbol{\Delta}\}$} \Comment{Go over every incoming arc to this vertex}

\State{$DP(v,\boldsymbol{\eta},i) \gets min \{DP(v,\boldsymbol{\eta},i),DP(u,\boldsymbol{\eta} - \bar{\boldsymbol{w}}(e) \cdot \boldsymbol{\Delta} ^{-1},i-1) + c(e)\}$.}

\Comment{get the least possible cost by going over potential edges connecting from other vertices. For previous vertices allow one less hop. Also ensure overall weight is within budget.}

\State{$PATH(v,\boldsymbol{\eta},i)$ stores the vertex whose entry was used to finally update $DP(v,\boldsymbol{\eta},i)$}.

\EndFor
\EndFor
\EndFor
\EndFor

\State \Return $DP(t,\lfloor (1+\eps)L \cdot \boldsymbol{\Delta}^{-1} \rfloor,n)$
\State \Return $PATH$

\end{algorithmic}
\end{algorithm}

\begin{claim} \label{cl:horvath_helper_1}
    When Algorithm \ref{alg:horvath_modified} terminates, for each node $v$, if $DP(v,\boldsymbol{\eta}, h)$ is not infinity, then there exists some $s\leadsto v$  path $P$ with $\leq h$ hops such that $\bar{\boldsymbol{w}}(P) \leq \boldsymbol{\eta} \cdot \boldsymbol{\Delta}$ and the cost of path P is $DP(v,\boldsymbol{\eta},h)$ (roughly, this shows that the algorithm is correct). 
\end{claim}

\begin{proof}
    We prove this by induction on hop count. The base case is true because when we allow zero hops, the cost to reach any vertex that is not the source is infinite (because it is impossible), while the cost to reach the source itself is zero for a non-negative pattern (and infinite for negative patterns). 
        
        To prove the inductive step, we will look at the update statement for $DP(v,\boldsymbol{\eta},h)$ in Algorithm \ref{alg:horvath_modified} (line 8). $DP(v,\boldsymbol{\eta},h)$ is set using some $DP(u,\boldsymbol{\eta} - \bar{\boldsymbol{w}}(e) \cdot \boldsymbol{\Delta}^{-1},h-1)$ where $u$ is a vertex with a $u \leadsto v$ edge. $DP(u,\boldsymbol{\eta} - \bar{\boldsymbol{w}}(e) \cdot \boldsymbol{\Delta}^{-1},h-1)$ allows strictly less than $h$ hops. Therefore it would have been examined and computed in a previous iteration. By induction, $DP(u,\boldsymbol{\eta} - \bar{\boldsymbol{w}}(e) \cdot \boldsymbol{\Delta}^{-1},h-1)$ is equal to the cost of an $s \leadsto u$ path $P_2$ with $\leq h-1$ hops such that $\bar{\boldsymbol{w}}(P_2)\leq \boldsymbol{\eta} \cdot \boldsymbol{\Delta} - \bar{\boldsymbol{w}}(e)$. Adding the edge $u-v$ to this, we have a path $P$ whose scaled weight is $\bar{\boldsymbol{w}}(P_2) + \bar{\boldsymbol{w}}(e) \leq \boldsymbol{\eta} \cdot \boldsymbol{\Delta}$ and thus our induction hypothesis is proved.
\end{proof}

\begin{claim} \label{cl:horavth_helper_2}
When Algorithm \ref{alg:horvath_modified} terminates, for each node $v$, if there is a $s \leadsto v$ path $P$ with $\leq h$ hops and $\bar{\boldsymbol{w}}(P) \leq \boldsymbol{\eta} \cdot \boldsymbol{\Delta}$, then $DP(v,\boldsymbol{\eta},h) \leq c(P)$ is satisfied (roughly, this shows that the algorithm is optimal).    
\end{claim}

\begin{proof}
    We will again use induction on hop count here. For the base case, the source has a path $P$ with zero hops and $\bar{\boldsymbol{w}}(P) \leq \boldsymbol{\eta} \cdot \boldsymbol{\Delta}$ for non-negative $\boldsymbol{\eta}$ and it doesn't have any such path for negative $\boldsymbol{\eta}$. The $DP$ values fulfill this rule. For any other vertex, there is no path with zero hops and thus all $DP$ values are set to zero for them.
        
    For the inductive case, if there is a $s\leadsto v $ path $P$ with $\leq h$ hops and $\bar{\boldsymbol{w}}(P) \leq \boldsymbol{\eta} \cdot \boldsymbol{\Delta}$, then there is some vertex $u$ such that $(u,v) \in E$ and there exists a $s \leadsto u$ path $P_2$ with $\leq h-1 $ hops. Let $\bar{\boldsymbol{w}}(P_2) =\bar{\boldsymbol{w}}(P) - \bar{\boldsymbol{w}}(e) $ be the weight of the path $P_2$ and let $c(P_2) = c(P) - c(e)$ be its cost. Our algorithm would have evaluated $DP(u,\bar{\boldsymbol{w}}(P_2) \cdot \boldsymbol{\Delta}^{-1},h-1)$ in a previous iteration (because it has $\leq h-1$ hops), and by the induction hypothesis this value would be $\leq c(P_2)$. Thus, after examining the edge $(u,v) = e$, our algorithm will store a cost $\leq DP(u,\bar{\boldsymbol{w}}(P_2) \cdot \boldsymbol{\Delta}^{-1},h-1) + c(e) \leq c{(P_2)} - c{(e)}  \leq c{(P)}$ and in addition the weight of the path returned by our algorithm would be $\leq \bar{\boldsymbol{w}(P_2)} + \bar{\boldsymbol{w}(e)} \cdot \boldsymbol{\Delta} = \bar{\boldsymbol{w}(P_1)}$.  
\end{proof}

\begin{lemma} \label{le:horvath_modified_main}
    If $|E|,n$ are the number of edges and vertices in the input graph $G$, then Algorithm \ref{alg:horvath_modified} runs in $O(|E| \cdot n \cdot |H|)$ time and returns an optimal and feasible path for \eqref{eq:horvath_scaled}. 
\end{lemma}

\begin{proof}

    For the runtime, see that we examine each edge once per hop and pattern. We have $n$ hops and $|H|$ patterns.

    Note that our algorithm does not need to consider any walk that is not a path. This is because, both the cost $c(e)$ and the scaled weight $\bar{\boldsymbol{w}(e)}$ cannnot have negative cycles. Therefore, if a walk that is used as a solution happens to include a cycle, then we can safely remove the cycle from the walk without violating feasibility (resource-wise) or increasing the cost. 
    
    The Lemma is true when the following two statements are true.

    \begin{enumerate}
        \item For each node $v$, if $DP(v,\boldsymbol{\eta}, h)$ is not infinity, then there exists some $s\leadsto v$  path $P$ with $\leq h$ hops such that $\bar{\boldsymbol{w}}(P) \leq \boldsymbol{\eta} \cdot \boldsymbol{\Delta}$ and the cost of path P is $DP(v,\boldsymbol{\eta},h)$. 
       
        \item For each node $v$, if there is a $s \leadsto v$ path $P$ with $\leq h$ hops and $\bar{\boldsymbol{w}}(P) \leq \boldsymbol{\eta} \cdot \boldsymbol{\Delta}$, then $DP(v,\boldsymbol{\eta},h) \leq c(P)$ is satisfied.
       
    \end{enumerate}
    
    Now, using Claim \ref{cl:horvath_helper_1} and Claim \ref{cl:horavth_helper_2} we see that the Lemma is proved.
   
\end{proof}

Recall Theorem \ref{thm:horvath_negative}. We now present its proof.

\thmhorvathnegative*

\begin{proof}
    From Lemma \ref{le:horvath_modified_main} we can solve \eqref{eq:horvath_scaled} in $O(|E| \cdot n \cdot |H|)$ time. This solution is both optimal and feasible. 

    Using Lemma \ref{le:horvath_conversion} we can use this to retrieve a solution for \eqref{eq:horvath_original}.

    Recall that $ \gamma_i:=|\min\{\min_{e \in E}\{w_i(e)\},0\}|/|L_i|$. Thus, $\gamma_i = |\maxni| / | L_i |$.
    
     Now, using Corollary \ref{cor:horvath_pattern_bound} we can see that $|H|$ is polynomial when $|\maxni| / | L_i |$ and $1/| \eps_i |$ are polynomials for all $i \in \{1,2,\ldots,m\}$. This proves our theorem.
\end{proof}

Also note that as in \cite{gkl2023}, when all the resource consumptions are integers polynomial in the graph size, we can use Algorithm \ref{alg:horvath_modified} to get a path of minimum cost that exactly satisfies the resource requirements. The idea is to just set $\eps_i$ to a small enough value so that any error that occurs from the usage of Algorithm \ref{alg:horvath_modified} is smaller than the smallest possible error for the given input graph. This gives us the following corollary.

\begin{corollary} \label{cl:horvath_modified_integer}
    When for all edges $e \in E$, $w_{e,i} \in [\poly(n)]_{\pm} \:\: \forall i \in \{1,2,\ldots,m\}$, there exists a fully polynomial time $(1;1,\ldots,1)-$ algorithm for $m-\rcsp$ that runs in time polynomial in input size.
\end{corollary}
\begin{proof}

    For the moment, let us assume all edge lengths are non-negative. Even if we allow edge lengths to be negative, the proof wouldn't have any significant changes. The smallest possible nonzero difference for the length between two different paths is $1$. This is because all the lengths are integers and therefore the total path length is also an integer. The smallest nonzero difference (in terms of absolute value) between two integers is $1$.

    The maximum length of any path is a polynomial in $n$. This is because we have $\leq n$ edges in a path and each of those edges has lengths that are polynomial in $n$. 

    When we set $\eps_i =  1/(n^2 \cdot \poly(n)) \leq 1/(n \cdot \poly(n))$ and run Algorithm \ref{alg:horvath_modified}, we can find a path with the exact resource requirements in polynomial time. Note that, because all edge lengths are integers polynomial in $n$,  $|\maxni| / (L_i)$ is a polynomial in $n$. Therefore using \ref{thm:horvath_negative} the claim is proved.
\end{proof}

We also present a slightly modified version of Corollary \ref{cl:horvath_modified_integer} in \ref{cl:horvath_modified_onerational} that can handle rational numbers in one resource but at the cost of going slightly over budget in that resource. 

Recall Corollary \ref{cor:rcsprational}. We now present its proof.

\corrcsprational*

\begin{proof}
    The proof is very similar to Claim \ref{cl:horvath_modified_integer}. We need a $\zeta_1$ that is small enough to ensure there is no error for the first $m-1$ resources. Then, we use $\zeta_2 = \min(\zeta_1,\zeta)$ to get a $(1;1+\zeta_2,\ldots,1+\zeta_2)$ path that fits all our requirements.  Note that the $m^{th}$ resource is always non-negative and thus  $|\maxni| / (L_i) = 0$ is a polynomial in $n$.
\end{proof}

\bibliographystyle{acm}
\bibliography{reference}

\end{document}